\documentclass[a4paper]{article}

\usepackage[english]{babel}
\usepackage[utf8]{inputenc}
\usepackage{graphicx}
\usepackage{amsmath}
\usepackage{amsthm}
\usepackage{amssymb}  
\newtheorem{theorem}{Theorem}
\newtheorem{lemma}[theorem]{Lemma}
\newtheorem{proposition}[theorem]{Proposition}

\newtheorem{definition}[theorem]{Definition}
\usepackage[position=top]{subfig}
\usepackage{tkz-euclide}
\usepackage{tikz}
\usetikzlibrary{trees,shapes,arrows}
\usepackage{scalefnt}
\usepackage{framed}
\usetikzlibrary{chains,shapes.multipart}
\usepackage{xcolor}
\usepackage{multirow}
\usepackage[margin=2cm]{geometry}

\usepackage{doi}
\usepackage[autostyle]{csquotes}
\usepackage[style=authoryear,citestyle=authoryear-comp,maxcitenames=2,maxbibnames=99,backend=bibtex,backref=true]{biblatex}
\usepackage{hyperref}
\hypersetup{
    pdftitle={\title{}},%
    bookmarksnumbered=true,bookmarksopen=true,
	unicode=true,colorlinks=true,linktoc=all,
	linkcolor=blue,citecolor=blue,filecolor=blue,urlcolor=blue,
	pdfstartview=FitH
}

\usepackage
	{todonotes}

\begin{document}
\title
	{Opinion-Based Centrality in Multiplex Networks: \\A Convex Optimization Approach}


\author{Alexandre \textsc{Reiffers-Masson} and Vincent \textsc{Labatut}\\
	Laboratoire Informatique d'Avignon (LIA) EA 4128 \\
    Universit\'e d'Avignon, France}


\maketitle
\begin{abstract}
Most people simultaneously belong to several distinct social networks, in which their relations can be different. They have opinions about certain topics, which they share and spread on these networks, and are influenced by the opinions of other persons. In this paper, we build upon this observation to propose a new nodal centrality measure for multiplex networks. Our measure, called \textit{Opinion centrality}, is based on a stochastic model representing opinion propagation dynamics in such a network. We formulate an optimization problem consisting in maximizing the opinion of the whole network when controlling an external influence able to affect each node individually. We find a mathematical closed form of this problem, and use its solution to derive our centrality measure. According to the opinion centrality, the more a node is worth investing external influence, and the more it is central. We perform an empirical study of the proposed centrality over a toy network, as well as a collection of real-world networks. Our measure is generally negatively correlated with existing multiplex centrality measures, and highlights different types of nodes, accordingly to its definition.

\vspace{0.3cm}
\noindent \textcolor{red}{\textbf{Cite as:} A. Reiffers-Masson \& V. Labatut. \href{https://www.cambridge.org/core/journals/network-science/article/opinionbased-centrality-in-multiplex-networks-a-convex-optimization-approach/05125ECC8229CF2E43248A96B4DFE5AF}{Opinion-based Centrality in Multiplex Networks}. Network Science, 5(2):213-234, 2017. Doi: \href{https://doi.org/10.1017/nws.2017.7}{10.1017/nws.2017.7}}
\end{abstract}

\section{Introduction}
\label{sec:Introduction} 
In our ultra-connected world, many people simultaneously belong to several distinct social networks, in which their relations can be different. For instance, interpersonal connections are not necessarily the same in Online Social Networks (OSN) and in the offline, real-world. Moreover, the way people are interconnected can even differ when considering several different OSNs such as Facebook, LinkedIn or ResearchGate, due to the fact these services have different purposes. The same can be said about the various types of interactions one can experience in the real world: it is possible for someone to maintain simultaneously several types of relationship with the same person (collaboration, kinship, friendship, etc.). This has consequences in terms of information propagation, since one can not only exchange information inside a single social network, but also receive some information through one network and fetch it via another one. This situation may occur, for instance, in the framework of Online Social Networks, when a user gets some news on Twitter and shares it on Facebook.

When identifying the central nodes in such a network, it is of course necessary to take the multilayer nature of the structure into account, in order to avoid any information loss. For this purpose, various measures have been proposed in the literature. Most of them generalize widespread existing unilayer measures such as degree \parencite{Magnani2011, Domenico2013, Battiston2014}, betweenness \parencite{Magnani2013,Sole-Ribalta2014, Chakraborty2016c}, closeness \parencite{Magnani2011, Sole-Ribalta2015}, Eigenvector \parencite{Domenico2013, Sola2013, Battiston2014}, PageRank \parencite{Ng2011, Halu2013, Coscia2013} or HITS \parencite{Kolda2006}. These generalizations rely on the adaptation of unilayer concepts to the multilayer case. For instance, measures based on matrix decomposition are modified to handle tensors, whereas in others, the concept of geodesic distance is altered to take inter-layer paths into account. Several recent articles review multilayer networks and the related centrality measures \parencite{Domenico2013, Kivelae2013, Boccaletti2014}.

In this article, we propose to generalize another type of uniplex approach, based on the resolution of optimization problems on diffusion models. One of the first opinion diffusion models was developed by \textcite{degroot1974reaching}. In this model, the persons are embedded in a social network and update their opinion (a real number in $[0,1]$) over time, by taking the average opinion of their neighbors. Some of the major extensions of this model are summarized by \textcite{jackson2008social}. Later works focused not only on understanding the opinion adoption process, but also on controlling it. Recent papers \parencite{borkar2010manufacturing, Bimpikis_competingover} suggest that we can use the theory of optimization and control in order to design efficient strategies for the control of opinion diffusion. In \parencite{borkar2010manufacturing}, the authors propose to impose a particular opinion to certain nodes in order to make the whole network adopt it too. In \parencite{Bimpikis_competingover}, the authors study how much a company needs to invest in a person in order to improve the adoption of some product it wants to sell to the whole network. The node targeted by these strategies can somehow be considered as central, and we want to explore this type of centrality in this work.

Our measure is designed specifically for multiplex networks. A multiplex network, or edge-colored multigraph in the nomenclature of \textcite{Kivelae2013}, corresponds to a multilayer network in which all nodes are present in all layers, and one node is connected to all its counterparts in all layers (the so-called categorical inter-layer coupling). Our approach is based on a model representing how the individual opinion, for a topic of interest, evolves among a group of persons. Each layer represents the influence people have on each other in a given social context or for a given social media. A person is influenced by his neighbors, like in DeGroot models. However, we introduce an additional influence in the model, which represents an external party able to affect each person individually. For example, in a marketing context, this external influencer could be a firm willing to direct its communication towards certain persons in the considered social group. Our centrality measure is related to the solution of the optimization problem consisting in determining which amount of external influence to invest in each person, in order to maximize the overall opinion of the social group.

Our contributions are the following. First, we propose a model of opinion diffusion for multiplex networks. Second, we define an optimization problem consisting in maximizing this diffusion when controlling the information entering the network. Third, we define a multiplex centrality measure based on the solution of this problem.

The rest of this article is organized as follows. In the next section, we describe our stochastic model of opinion evolution. In Section \ref{sec:centrality}, we define our optimization problem, and use it to derive our centrality measure. In Section \ref{sec:experiments}, we study its behavior on a toy network and on real-world networks, comparing it to existing multiplex centrality measures.

\section{Opinion Dynamics: A Stochastic Model}
\label{sec:model}
We propose to model the evolution of the opinion of a group of people belonging to different social networks. In this section, we describe our model, before using it to define a multiplex centrality measure in the next section.

Let us consider a group of people, some of which have influence over others. This can be represented as a network, whose nodes correspond to each individual and directed links correspond to the influence one individual exerts over another. In the well-known Degroot model \parencite{degroot1974reaching, jackson2008social}, the opinion of a person regarding a topic of interest depends on the opinion of the people influencing him. The work presented in \parencite{Bimpikis_competingover} is based on this model, with an additional controlled exogenous influence, able to affect each individual. The authors use this modified model to identify optimal control strategies in terms of opinion diffusion. Such a strategy consists in determining which individual to influence, and how much they should be influenced. Our approach is based on the same principle, except we extend the model by allowing several media for inter-individual influence.

Suppose we want to consider not one, but \textit{several} social networks at once. This results in a multiplex graph, in which each layer corresponds to the influence relationships observed in a specific social network (layer). The nodes are the same in all layers, since they represent the same person in various contexts. Formally, we note $\mathcal{I}:=\{1,\ldots,I\}$ the set of persons and $\mathcal{C}:=\{1,\ldots,C\}$ the set of social networks. For a given social network $c\in\mathcal{C}$ (layer), we call $E_c\in[0,1]^{I\times I}$ the \textit{imitation matrix}, whose $ij$-entry $e_{ijc}$ is the probability that person $j$ is influenced by person $i$ in social network $c$, and consequently mimics his opinion. In other words, $E_c$ is the (normalized) weighted adjacency matrix of the $c^{th}$ layer of our multiplex network.

Let us now describe the process taking place on this network. We assume that a person $i\in\mathcal{I}$ is influenced by a source external to the network, according to a Poisson point process of intensity $\lambda_{i}\in\mathbb{R}_+$. Moreover, a person $i$ decides to mimic the opinion of one of his neighbors in layer $c$ according to a Poisson point process of intensity $\alpha_{ic}\in\mathbb{R}_+$. Let $k\in\mathbb{N}_+$ denote the $k^{th}$ event (external or internal influence). We define $\Lambda:=\sum_{j=1}^I\sum_{d=1}^C\lambda_{j}+\alpha_{jd}$. Then, $\lambda_{i}\Lambda^{-1}$ is the probability that the $k^{th}$ event is individual $i$ undergoing an external influence, whereas $\alpha_{ic}\Lambda^{-1}$ is the probability for this event to be individual $i$ mimicking his neighbors' opinion in social network $c$ (i.e. undergoing an internal influence). 

For each event $k$, let $x_{i}(k)\in[0,1]$ be the opinion of person $i$. A zero value means person $i$ has no interest for the considered topic, whereas $1$ represents a full interest. By modeling opinions using real values in $[0,1]$, we follow the line of work of \textcite{degroot1974reaching}. We call the vector $\mathbf{x}(k):=(x_{1}(k)\ldots,x_{I}(k))$ the \textit{opinion profile} of the whole social group at event $k$. As mentioned before, it can be updated due to the direct influence of other individuals' opinions, or due to an external influence.  For each $i$, let $x_{i}(0)=x_{i}^0\in[0,1]$. 

Let us now detail how the opinion of a person $i$ is updated. One of three scenarios are possible:
\begin{enumerate}
	\item If at event $k$, person $i$ receives a message (with probability $\lambda_{i}\Lambda^{-1}$), then his opinion is updated as follows:
	\begin{eqnarray}
		x_{i}(k+1)&=&x_{i}(k)(1-\delta)+\delta  
	\end{eqnarray}
	where $\delta\in]0,1[$ is the step-size. 
    
	\item If at event $k$, person $i$ decides to imitate its neighbors in social network $c$ (which occurs with probability $\alpha_{ic}\Lambda^{-1}$), then his opinion is updated as follows:
	\begin{eqnarray}
		x_{i}(k+1)&=&x_{i}(k)(1-\delta)+\delta\left(\sum_{j=1}^Ie_{jic}x_{j}(k)\right),	\end{eqnarray}

	\item If event $k$ does not concern person $i$ (which occurs with probability $(\sum_{c}\sum_{j\neq i}\alpha_{ic}+\lambda_j)\Lambda^{-1}$), then his opinion is updated as follows:
	\begin{eqnarray}
		x_{i}(k+1)&=&x_{i}(k)(1-\delta).
	\end{eqnarray}
\end{enumerate}

In this model, opinion evolution is a constant step-size stochastic approximation. Indeed, for each $k$ and each person $i$, $x_{i}(k)$ can be rewritten as follows:
\begin{eqnarray}
	\label{eq:discreteevolutionofx}
	x_{i}(k+1)&=& x_{i}(k)+\delta \left(Y_{ic}(k)-x_{i}(k)\right), \\\nonumber
	&&x_{i}(0)=x_{i}^0,
\end{eqnarray}
where 
\begin{equation}
	\label{eq:updateuser}
	Y_{ic}(k):=\left\{\begin{array}{lll}
	1&\text{ w.p }& \lambda_{i}\Lambda^{-1},\\
	\displaystyle{\sum_{j=1}^Ie_{jic}x_{j}(k)}&\text{ w.p } &\alpha_{ic}\Lambda^{-1},\\
	0 &\text{ w.p } & (\sum_{c}\sum_{j\neq i}\alpha_{ic}+\lambda_j)\Lambda^{-1}.
	\end{array}\right. 
\end{equation} 

The previous equations suggest that for each person $i$,  $x_{i}(k)$ is a stochastic finite difference Euler schemes of the following system of differential equations:
\begin{eqnarray}
	\label{eq:diffeq}
	\dot{x}_{i}(t)&=&\lambda_{i}\Lambda^{-1}+\Lambda^{-1}\sum_{c=1}^C\alpha_{ic}\sum_{j=1}^Ie_{jic}x_{j}(t)-x_{i}(t),  \\\nonumber
	&&x_{i}(0)=x_{i}^0.
\end{eqnarray}
Let $\overline{E}$ be a matrix such that the $ij$-entry of $\overline{E}$ is equal to $\overline{e}_{ij}=\sum_{c=1}^{C}\alpha_{ic}e_{jic}$. Let $E^T$ be the transpose of matrix $E$. Let $A$ be equal to $(\overline{E}-Id_{I})^{-1}$ where $Id_I$ is the identity matrix of dimension $I^2$.

We are interested to compute, if it is possible, $\lim_{k\rightarrow +\infty}x_{i}(k)$. We use the theory of stochastic approximation \parencite{borkar2008stochastic}, which highlights the relation between $\lim_{k\rightarrow +\infty}x_{i}(k)$ and $\lim_{t\rightarrow +\infty}x_{i}(t)$. The next proposition gives us the expected result.

\begin{proposition}
	\label{proposition:convergence} 
    If for all $i$, $\sum_{j=1}^I\overline{e}_{ij}<1$, $\frac{1}{2}[\overline{E}+\overline{E}^T]$ has negative Eigenvalues and $\delta \ll 1$, then for each $i$,
\begin{eqnarray}
	\lim_{k\rightarrow +\infty}x_{i}(k)&=&x_{i}^*,
    \end{eqnarray}
where $x_{i}^*$ is solution of
\begin{equation}
0=\lambda_{i}\Lambda^{-1}+\Lambda^{-1}\sum_{c=1}^C\alpha_{ic}\sum_{j=1}^Ie_{jic}x_{j}^*-x_{i}^*.\end{equation}
\end{proposition} 
  
\begin{proof}[Proof of Proposition \ref{proposition:convergence}]
	Our proof is constituted of three steps. In the first, we rewrite $\textbf{x}(n)$ as a stochastic approximation. In the second, we prove the uniqueness of the rest point. In the third, the stability of the unique rest point is proved by using \textcite[][chap.10, p.131]{borkar2008stochastic}. We conclude by using Theorem 3 of \textcite[][p.106]{borkar2008stochastic}.

\textbf{Step 1: } For each $i$, the evolution of $x_{i}(k)$ is equivalent to the following stochastic approximation: 
\begin{eqnarray}\nonumber
x_{i}(k+1)&=&  x_{i}(k)+\delta\left(Y_{i}(k)-x_{i}(k)\right)\\\nonumber
&=& x_{i}(k)+\delta\left(E[Y_{i}(k)])-x_{i}(k)\right)+\delta\left(Y_{i}(k)-E[Y_{i}(k)])\right).
\end{eqnarray}
We observe that for each person $i$:  
$$
M_{i}(n) :=Y_{i}(k)-E[Y_{i}(k)]$$
is a martingale difference sequence of zero mean (a definition can be found in the appendix of \textcite{borkar2008stochastic}). Thus, for each $i$, 
$x_{i}(k)$ is a stochastic approximation.\\

\textbf{Step 2:} From the fact that $x_{i}(k)$ is a stochastic approximation, we study the asymptotic behavior of the associated differential equations:
\begin{eqnarray}
	\label{eq:diffeq1}
	\dot{x}_{i}(t)&=&\lambda_{i}\Lambda^{-1}+\Lambda^{-1}\sum_{c=1}^C\alpha_{ic}\sum_{j=1}^Ie_{jic}x_{j}(t)-x_{i}(t),  \\\nonumber
	&&x_{i}(0)=x_{i}^0.
\end{eqnarray}

Rest points of \eqref{eq:diffeq1} are the solutions of the following linear system:
\begin{equation}
	\label{eq:rest.point.fast.equation}
	\lambda\Lambda^{-1}+(\overline{E}-Id_I)\mathbf{x}=0
\end{equation}
where the $ij$-entry of $\overline{E}$ is equal to $\overline{e}_{ij}=\sum_{c=1}^{C}\alpha_{ic}e_{jic}$ and $Id_I$ is the identity matrix of dimension $I^2$.
We notice that $\overline{E}-Id_I$ is a strictly diagonally dominant matrix \parencite{horn2012matrix}. This is the reason why $\overline{E}-Id_I$ is invertible. We can conclude that \eqref{eq:rest.point.fast.equation} has a unique rest point given by:
\begin{equation}
	\label{eq:fast.rest.point}
	\mathbf{x}^*=(\overline{E}-Id_I)^{-1}(-\lambda\Lambda^{-1}).
\end{equation}

\textbf{Step 3:} From the fact the matrix $\frac{1}{2}[(\overline{E}-Id_I)+(\overline{E}-Id_I)^T]$ has negative Eigenvalues, $\mathbf{x}^*$ is the unique stable point of \eqref{eq:diffeq1} \parencite[see][chap.10, p.131]{borkar2008stochastic}.

Finally, we apply theorem 3 p.106 of  \cite{borkar2008stochastic} to show that $\textbf{x}(n)$ (whose components are given by (\ref{eq:discreteevolutionofx})) converges to $\textbf{x}^*$.
\end{proof}
  
 In this section, we developed a framework that models the evolution of individual opinion occurring on a multiplex social networks. We provided a mathematical expression of the asymptotic behavior of opinion dynamics. On this base, in the next section, we define the multiplex opinion problem and derive a multiplex centrality measure.

\section{An Opinion-Based Multiplex Centrality}
\label{sec:centrality}
In our model, the parameter $\lambda_{i}$ determines how much external influence person $i$ receives. Our centrality measure is based on the control of this parameter. We assume that the centrality of a person in the whole multiplex network depends on the amount of external influence one should exert on him in order to increase the global opinion level of the whole social group. According to this statement, finding the most central person amounts to finding the person whose stimulation (through external influence) maximizes the total opinion. So, in order to process our centrality, we need first to solve an optimization problem. In the rest of this section, we describe the problem and derive a mathematical closed form of the solution.

\subsection{Opinion Maximization Problem}
\label{sec:problem}
We want our centrality to quantify how important an individual is, in term of resource allocation. For this purpose, we need to maximize the opinion profile $\mathbf{x}$ by controlling the amount of external influence $\lambda:=(\lambda_{1},\ldots,\lambda_{I})$ spent on the individuals of the social group. We propose a first centrality measure, which we call \textit{Naive Opinion Centrality}, and define it as follows:
 \begin{definition} Let $U(\mathbf{x}^*(\lambda)):=\sum_{i=1}^Ix_{i}^*(\lambda)$ the utility function with which we compute the opinion centrality (note it may also be non-linear: we discuss this and provide examples later). The \textbf{Naive Opinion Centrality} is the vector $\lambda^{NO}$ which is the optimum of the following problem:
   \begin{equation}
		\label{eq:multiplex opinion problem}
		\max_{\lambda\geq \mathbf{0}}U(\mathbf{x}^*(\lambda))
	\end{equation}
	where, for each $i$ and $c$, $x^*_{i}(\lambda)$ is solution of:
	\begin{equation}
    	\label{eq:constraintopinion1}
		\lambda_{i}\Lambda^{-1}+\Lambda^{-1}\sum_{c=1}^C\alpha_{ic}\sum_{j=1}^Ie_{jic}x^*_{j}(\lambda)-x_{i}(\lambda)=0,
	\end{equation}
	and the total intensity of the external influence is constrained by a so-called budget $R>0$:  
	\begin{equation}
    	\label{eq:budgetconstraint}
		\sum_{i=1}^I\lambda_{i}= R.
	\end{equation}
	We call this optimization problem the \textbf{Opinion Maximization Problem (OMP)}.
\end{definition}

According to the previous definition, our centrality defines an index, per person, which is equal to the amount of resources that an external influencer needs to spend on each person such that he will maximize the opinion of the whole social group, regarding some topic or product of interest. Therefore, if a node is twice as central as another, this can be interpreted as the need to invest twice the amount of resources in this node in order to obtain optimal diffusion. In addition to creating a ranking measure, the opinion centrality can also be used to design a targeting strategy. The budget $R$ is used to prevent the external influencer from targeting every person in the social group, which would lead to a degenerate situation. 

This paper is not the first one making the relation between targeting strategies and centrality measures. In \parencite{borkar2011controlled, borkar2010manufacturing, Bimpikis_competingover}, the authors maximize the opinion in a uniplex social network, by spending resources on each person. Our paper has two majors differences with these works. The first one is that, in each of these papers, the authors assume that the persons are embedded in only one social network, and not several ones as proposed in our framework. The second difference is that the authors only mention a connection to the notion of centrality, but never study it as such.

We now provide an interpretation of the mathematical expression of the utility $U(\mathbf{x}):=\sum_{i=1}^Ix_{i}^*(\lambda)$. Behind it, there is the assumption that the opinions of distinct persons are perfect substitutes, i.e. can be exchanged without information loss. More precisely, an influencer does not distinguish the opinion of two distinct persons, and so an increase in the opinion of one of them decreases his interest for the other. This hypothesis is expressed through a linear utility $U(\mathbf{x})$, i.e. we model the network opinion as the sum of all individual opinions. It could be considered as strong, however our whole method can be extended by using other non linear utility functions, such as: 
\begin{eqnarray} 
    \label{eq:complementary topic}
	U(\mathbf{x})&=&\min_{i\in\{1,\ldots,I\}} w_{i}x_{i},\\   
    \label{eq:cobb douglas}
	U(\mathbf{x})&=&\prod_{i=1}^Ix_{i}^{w_{i}}
\end{eqnarray}
 
Both functions \eqref{eq:complementary topic} and \eqref{eq:cobb douglas} are related to the theory of utility function in economics \parencite{mas1995microeconomic}. The first, Eq. \eqref{eq:complementary topic}, can be used to describe perfect complements between opinions. It allows modeling the network opinion as the smallest individual opinion. The second, Eq. \eqref{eq:cobb douglas}, is the Cobb-Douglas function, which is a good trade-off between complement and substitute effects: the network opinion is represented as the product of individual opinions.


\subsection{Opinion Centrality Measure}
In the previous subsection, we define the Naive Opinion Centrality and give an interpretation of this measure. However, we need to ask the following question: Is the opinion centrality a relevant ranking measure? The answer is negative, and it is not hard to see that when considering the following lemma:
\begin{lemma}
	\label{lem:expression.naive.opinion} 
    If $\max_{j} \{-\sum_{i=1}^Ia_{ij}\}$ is unique and if for all $i$, $\sum_{j=1}^I\overline{e}_{ij}<1$, then:
	\begin{equation}
		\lambda_j^{NO}=\left\{\begin{array}{lll}
			R&\text{if }j=\text{argmax }\{-\sum_{i=1}^Ia_{ij}\},&\\
			0&\text{otherwise.}&
		\end{array}\right.
	\end{equation}
\end{lemma}

\begin{proof}[Proof of Lemma \ref{lem:expression.naive.opinion}]
Because the matrix $A$ is diagonally dominant, the solution of \eqref{eq:constraintopinion1} is given by:
\begin{equation}
	\mathbf{x}^*=A(-\lambda\Lambda^{-1}). 
\end{equation}
Thus we can rewrite the Naive Opinion Centrality \eqref{eq:multiplex opinion problem} as follows:
   \begin{equation}
   		\label{eq:multiplex opinion problem 2}
		\max_{\lambda\geq \mathbf{0}}U(\mathbf{x}(\lambda))
        := -\Lambda^{-1}\sum_{i=1}^I\sum_{j=1}^Ia_{ij}\lambda_{j},
\end{equation}
subject to:
\begin{equation}
		\label{eq:budgetconstraint2}
		\sum_{i=1}^I\lambda_{i}= R.
\end{equation}

$\Lambda^{-1}$ is constant and this implies that \eqref{eq:multiplex opinion problem 2} is a linear program. The conclusion follows by using the first order optimality condition \parencite{boyd2004convex}.
\end{proof}

According to this lemma, the naive opinion centrality, $\lambda_j^{NO}$, will always be equal to $0$ if $j$ is different from $\text{argmax }\{-\sum_{i=1}^ Ia_{ij}\}$ and so it is useless in term of ranking. The intuition behind this result is that the opinion strategy is motivated by finding the best person to target. However, in the context of centrality measures, one aims at ranking all the nodes, rather than identifying a single one fulfilling some constraints. Therefore comes the question: Can we modify the OMP such that the associated opinion centrality is appropriate for ranking? We can propose a positive answer to this question, inspired by machine learning theory. For this purpose, we make use of a regularization term, which allows us to define our final measure, which we call \textit{Opinion Centrality}. 

\begin{definition} The \textbf{Opinion centrality} is the vector $\lambda^{O}$ which is the optimum of the following problem:
   \begin{equation}
   		\label{eq: regularized multiplex opinion problem}
		\max_{\lambda\geq \mathbf{0}}U^R(\mathbf{x}^*(\lambda)):=\sum_{i=1}^Ix^*_{i}(\lambda) \quad -\underbrace{\frac{\gamma}{2}\sum_{i=1}^I(\lambda_i)^2}_{\text{Regularization function}}
	\end{equation}
	where $\gamma\in\mathbb{R}_+$ is the regularization coefficient, and the total intensity of the external influence is constrained by $R>0$:  
	\begin{equation}
    	\label{eq:budget constraint ROMP}
		\sum_{i=1}^I\lambda_{i}= R.
	\end{equation}
	We call this optimization problem the \textbf{Regularized Opinion Maximization Problem (ROMP)}.
 \end{definition}

We are now able to compute a mathematical expression of the opinion centrality. 
\begin{proposition}
	\label{prop:expression.opinion.centrality}
	If for all $i$, $\sum_{j=1}^I\overline{e}_{ij}<1$, and if $\gamma$ is such that
	\begin{equation}
    	\label{eq:positivity}
		\gamma>(\Lambda R)^{-1}I^2(\max\{a_{ij}\}-\min\{a_{ij}\}),
	\end{equation}
	then for all $j$
	\begin{equation}
		\lambda_j^{O}=RI^{-1}+\gamma^{-1}I^{-1}\Lambda^{-1}\sum_{i=1}^I\sum_{j=1}^Ia_{ij}-\gamma^{-1}\Lambda^{-1}\sum_{i=1}^Ia_{ij}.
	\end{equation}
\end{proposition}

\begin{proof}[Proof Proposition \ref{prop:expression.opinion.centrality}]
Because the matrix $A$ is diagonally dominant, the solution of \eqref{eq:constraintopinion1} is given by:
\begin{equation}
\mathbf{x}^*=A(-\lambda\Lambda^{-1}). 
\end{equation}
Thus we can rewrite the Regularized Opinion Centrality \eqref{eq: regularized multiplex opinion problem} as follows:
   \begin{equation}
   		\label{eq:multiplex opinion problem 3}
		\max_{\lambda\geq \mathbf{0}}U^R(\mathbf{x}^*(\lambda)):=-\Lambda^{-1}\sum_{i=1}^I\sum_{j=1}^Ia_{ij}\lambda_{j}-\frac{\gamma}{2}\sum_{i=1}^I(\lambda_i)^2,
	\end{equation}
subject to:
\begin{equation}
		\label{eq:budgetconstraint1}
		\sum_{i=1}^I\lambda_{i}= R.
	\end{equation}
The first order optimality condition \parencite{boyd2004convex} is as follows:
It exists $\beta\in \mathbb{R}$ such that $\beta$ and $\mathbf{x}^O$ are solution of the following system:
\begin{eqnarray}
	-\Lambda^{-1}\sum_{i=1}^Ia_{ij}-\gamma\lambda_i^O-\beta&=&0,\\
	\sum_{i=1}^I\lambda_{i}^O= R&=&0,
\end{eqnarray}
which is equivalent by taking the sum over $i$ in the first equation, to,
\begin{eqnarray}
	-\Lambda^{-1}\sum_{i=1}^Ia_{ij}-\beta&=&\gamma\lambda_i^O,\\
	-\Lambda^{-1}\sum_{i=1}^I\sum_{j=1}^Ia_{ij}-I\beta&=&\gamma R.
\end{eqnarray}
Now we can deduce that:
\begin{equation}
	-\beta=I^{-1}\gamma R+I^{-1}\Lambda^{-1}\sum_{i=1}^I\sum_{j=1}^Ia_{ij}.
\end{equation}
and by injecting $-\beta$ in the previous equations we get:
\begin{equation}
	\lambda_i^O= RI^{-1}+\gamma^{-1}I^{-1}\Lambda^{-1}\sum_{i=1}^I\sum_{j=1}^Ia_{ij}-\gamma^{-1}\Lambda^{-1}\sum_{i=1}^Ia_{ij}\overset{\eqref{eq:positivity}}{>}0,
\end{equation}
which concludes the proof.
\end{proof}

Even if we find a mathematical expression of the opinion centrality, it is still hard to interpret the impact of $\alpha$, $R$ and $\gamma$ over $\lambda^O$. First we observe that for each $j$ and $j'$:
\begin{equation}
	\mid\lambda_j^{O}-\lambda_{j'}^{O}\mid=\gamma^{-1}\Lambda^{-1}\mid\sum_{i=1}^Ia_{ij}-\sum_{i=1}^Ia_{ij'}\mid.
\end{equation}
Therefore the difference between two opinion centrality values will increase the more $\gamma$ and $R$ decrease ($R$ appears in $\Lambda^{-1}$). 
Despite this fact, the ranking is independent of $\gamma$ and $R$.

Concerning $\alpha$, we propose to study its behaviour in a particular context. Consider, for each $i$ and each $c$,  $\alpha_{ic}=\hat{\alpha}$. Then for each $i$, $x_i^*(\lambda)$ is solution of 
\begin{eqnarray}
	& & \lambda_{i}(R+CI\hat{\alpha})^{-1}+(R+CI\hat{\alpha})^{-1}C\hat{\alpha}\sum_{j=1}^Ie_{jic}x_{j}^*(\lambda)-x_{i}^*(\lambda)=0,\\
	& \Leftrightarrow& \lambda_{i}+C\hat{\alpha}\sum_{j=1}^Ie_{jic}x_{j}^*(\lambda)-(R+CI\hat{\alpha})x_{i}^*(\lambda)=0,\\
	& \overset{\hat{\alpha} \rightarrow +\infty}{\Rightarrow}& \lambda_{i}-(R+CI\hat{\alpha})x_{i}^*(\lambda)=0,\\
	& \Rightarrow& \lambda_{i}(R+CI\hat{\alpha})^{-1}=x_{i}^*(\lambda),
\end{eqnarray}
and we can conclude that 
\begin{equation}
	\lambda_i^O=RI^{-1}.
\end{equation}
In other words, in this context, and for a large $\hat{\alpha}$, the optimal strategy consists in investing an equal amount of external influence in all nodes. This is due to the fact the external influence becomes negligible compared to the internal influence.

\section{Experimental Validation}
\label{sec:experiments}
To assess the behavior of the opinion centrality, we consider first a toy problem designed to study some cases of particular interest for our measure (Section \ref{sec:ToyPb}), and then a collection of real-world multiplex networks (Section \ref{sec:RealWorld}). One could compare the measures simply based on the values they produce. However, centrality measures are generally used to compare nodes \textit{within} a single network, in which case the values themselves are not relevant: one should consider their rank instead \parencite{Sola2013}. We consequently decided to emphasize this aspect in our experimental assessment.

\subsection{Toy Problem}
\label{sec:ToyPb}
We first propose to explore the effect of the different parameters over a specific  multiplex network, called the \textit{barrel network}. The original version is in fact uniplex, constituted of two star networks connected by their centers \parencite{Bimpikis_competingover}. We obtain a multiplex network simply by replicating the same structure to form several layers, as shown in Figure~\ref{fig:barrelnetwork}. 

\begin{figure}[htb]
	\center
	\includegraphics[width=0.60\textwidth]{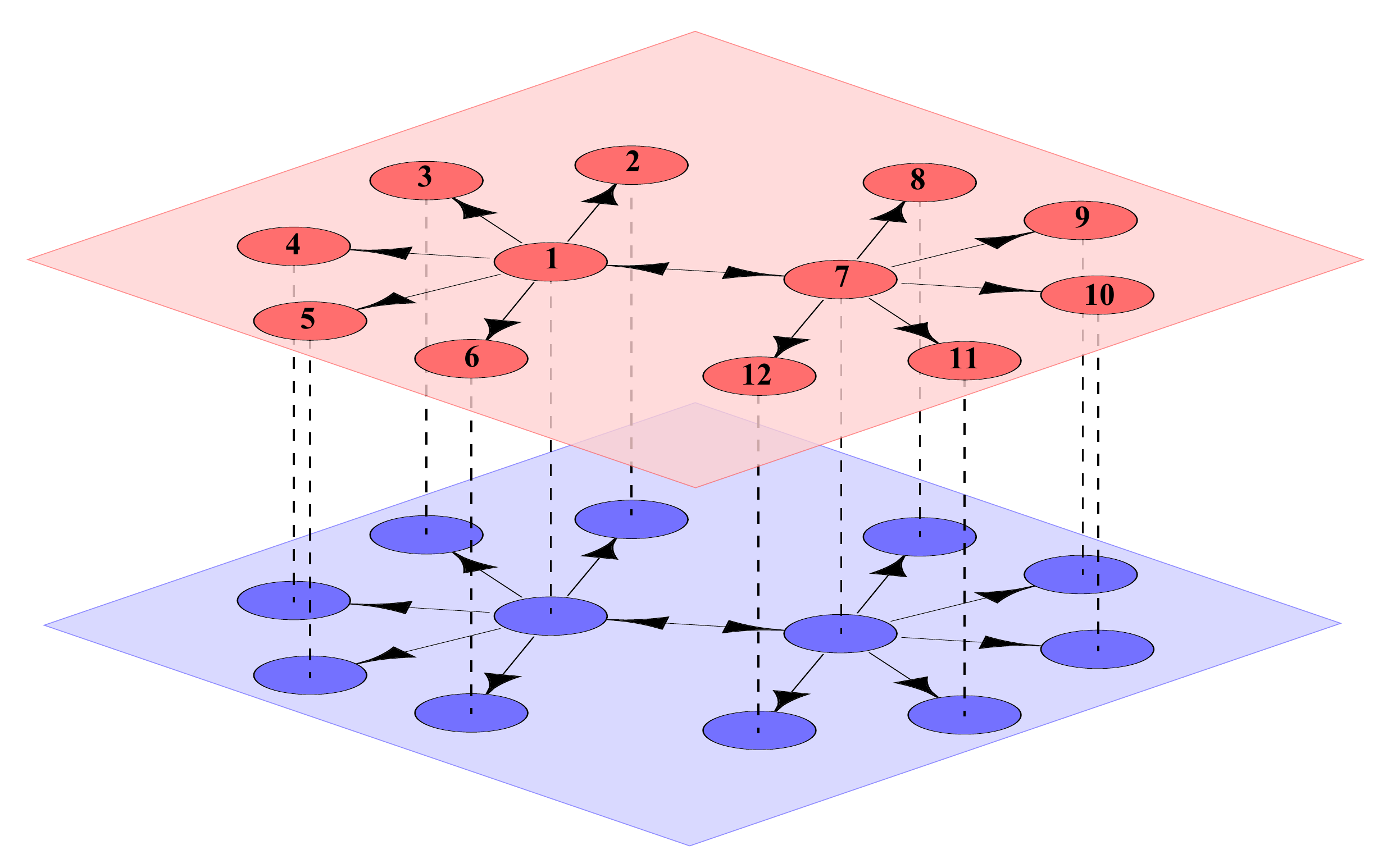}
	\caption{Example of a multiplex barrel network containing $12$ nodes and $2$ layers.}
	\label{fig:barrelnetwork}
\end{figure}

The centers, which correspond to nodes $1$ and $I/2+1$, are called \textit{Hubs}, whereas the rest of the network, i.e. nodes $2,..., I/2,I/2+2,..., I$, are the \textit{Leaves}. We assume that for each $i$ and each $c$, the $ij$-entry of matrix $E_c$ is given by:
\begin{equation}
	e_{jic}=\left\{\begin{array}{lll}
		& e_{0c} &\text{ if } (i,j)=(1,I/2+1) \text{ or } =(I/2+1,1)\\ 
		& e_{1c} &\text{ if } j=1 ,\; i\in\{2,\ldots,I/2\}\\
		& e_{2c} &\text{ if } j=I/2+1 ,\; i\in\{I/2+2,\ldots,I\}\\
		& 0 & \text{otherwise}.
		\label{eq:ToyModelEdges}
	\end{array}\right.
\end{equation}
In other words, for a given layer $c$, we consider three types of edges: $e_{0c}$ corresponds to the weight of the edges connecting the hubs, whereas $e_{1c}$ and $e_{2c}$ are the hub-to-leaf weights for the two stars constituting the layer.

We define $e_j'=\Lambda^{-1}\sum_{c=1}^C\alpha_{ic}e_{jc}$ for each $j\in \{0,1,2\}$. This variable can be seen as the average effect of a specific type of edge (as defined in \eqref{eq:ToyModelEdges}) in the overall multiplex network. In this particular case, we can get a mathematical closed expression of the opinion degree and PageRank centrality measures (see Table~\ref{tab:cent barrel}). This computation can be done independently from matrix $\alpha$. We compute the degrees and the PageRank centrality over the weighted aggregated network, which is obtained by collapsing the layers to get a uniplex network. The computation of the opinion centrality is made without taking into account $R$, with $\gamma=1$ and by dividing by $\Lambda$. This will not change the order of the nodes, since the budget (that can be seen as a normalization constraint) and the factor $\gamma$ do not affect node ordering. 

\begin{table}[htb]
	\center
	\begin{tabular}{l c c c c}
		\hline
		Node & \;Opinion centrality &\;Out-degree centrality & \;In-degree centrality & \;PageRank  \\
		\hline
		Hub $1$& $\frac{1+e_0'(1+e_2')+e_1'}{(1-e_0')^2}$ & $e_0'+e_1'$ & $e_0'$ & $\frac{1}{I {\left(5-4 \, e_{0}' \right)}}$\\
		Hub $I/2+1$\; & $\frac{1+e_0'(1+e_1')+e_2'}{(1-e_0')^2}$ & $e_0'+e_2'$ & $e_0'$ & $\frac{1}{I {\left(5-4 \, e_{0} \right)}}$\\
		Leave $2,\cdots, I/2$ & 1 & 0 & $e_1'$ & $\frac{-4 \, e_{0}' + 4 \, e_{1}' +5}{5 \, I {\left(5-4 \, e_{0}' \right)}}$ \\
		Leave $I/2+2,\cdots, I$ & 1 & 0 & $e_2'$ & $\frac{-4 \, e_{0}' + 4 \, e_{2}' + 5}{5 \, I {\left(5-4 \, e_{0}' \right)}}$ \\
		\hline
	\end{tabular}
	\caption{Expressions of the centrality measures for the barrel network.}
	\label{tab:cent barrel}
\end{table}

The first observation is that Leaves have always an opinion centrality equal to $1$, independently from any $e_j'$. The same can be said from the out-degree, except the Leaves have a centrality of $0$ instead. It is one major difference with both other centrality measures, in particular with PageRank, where each node including the Leaves depends on certain $e_j'$. The second observation is about the nature of the influence of the $e_j'$ over the opinion centrality and PageRank. The opinion centrality of the Hubs increases in $e_0'$, $e_1'$ and $e_2'$, and their PageRank increases in $e_0'$. Concerning the Leaves, as mentioned before their opinion centrality is not affected by $e_j'$, however their PageRank decreases in $e_0'$ and increases in $e_1'$ and $e_2'$. This highlights the fact the opinion measure behavior differs from the other measures', especially PageRank.

\begin{table}[htb]
	\center
	\begin{tabular}{l c c c c}
		\hline
		Scenario & \;Opinion centrality & \;Out-degree centrality & \;In-degree centrality & \;PageRank  \\
		\hline
		$e_0'\approx 0$, $e_1'\approx e_2'$\; & \eqref{eq:order1} & \eqref{eq:order} & \eqref{eq:order1}   & \eqref{eq:order1}\\
		$e_0'\approx 1$, $e_1'\approx e_2'$ & \eqref{eq:order1} & \eqref{eq:order} & \eqref{eq:order} & \eqref{eq:order}\\
		\hline
	\end{tabular}
	\caption{Node ranks for the considered centrality measures, in two different scenarios.}
	\label{tab:order barrel}
\end{table}

In order to understand better how the considered centrality measures differ in terms of node ranking, we investigate two scenarios, as described in Table~\ref{tab:order barrel}. In the first one we consider that Hubs are not linked together ($e_0'\approx 0$). In the second one, the inverse situation occurs and Hubs are strongly connected to each other ($e_0'\approx 1$). In order to simplify the computation, in both scenario we assume that $e_1'\approx e_2'$. In these cases, one can obtain one of two possible rankings:
\begin{equation}
	\label{eq:order}
	Hub \;1 \approx Hub \;I/2+1 > Leave\; 2,\cdots, I/2,I/2+2,\cdots, I,
\end{equation}
\begin{equation}
	\label{eq:order1}
	Hub \;1 \approx Hub \;I/2+1 < Leave\; 2,\cdots, I/2,I/2+2,\cdots, I.
\end{equation}

Besides the scenarios, Table~\ref{tab:order barrel} also shows which node ranking each measure provides. Out-degree behaves similarly in both cases, because the out-degree of Leaves is not influenced by the $e_j'$. For the other measures, on the contrary, there is a switch in ranking when changing the scenario. The in-degree and PageRank measures behave similarly, switching from the second \eqref{eq:order1} to the first \eqref{eq:order} ranking when $e'_0$ increases. This is due to the fact $e'_0$  captures the connection between nodes which is crucial in PageRank and in-degree centrality, as can be observed in their mathematical expression. The opinion centrality does not undergo any switch: the more $e'_0$ increases, the more it is important to invest in the hub to maximize the opinion propagation, even in the case when the hubs are weakly connected. In the considered network, the opinion centrality is negatively correlated (in terms of rank correlation) with the out-degree.

As explained before, our model requires to define a parameter $\alpha$ for the opinion measure, which represents the amount of influence internal to the network (by opposition to the external influence, which is determined when solving the optimization problem described in Section \ref{sec:problem}). This $\alpha$ parameter is an $I \times C$ matrix, where $\alpha_{ic}$ corresponds to user $i$ in layer $c$.
Here, to simplify our analysis, we consider the previously introduced $\hat{\alpha}$, where for each $i$ and each $c$, $\alpha_{ic}=\hat{\alpha}$. We are interested in understanding the evolution of the opinion centrality in a $2$-layer $12$-node barrel network similar to that depicted in Figure \ref{fig:barrelnetwork}, when $\hat{\alpha}$ grows. Figure~\ref{fig:centralitybarrelnetwork} displays this evolution for the different types of nodes, as a function of $\hat{\alpha} \in\{1,\ldots,100\}$, with $R=1$, $e_{0c}=0.1$, $e_{1c}=0.2$ and $e_{2c}=0.3$ for all $c$. We observe that the value of the opinion centrality per node converges to $0.08332089\approx\frac{1}{12}$. It therefore seems that we converge to a uniform strategy, in which each user receives the same amount of external influence. This observation confirms the result proved in the previous section. Moreover, we observe that the growth of $\hat{\alpha}$ does not affect the node ranks.

\begin{figure}[htb]
	\center
    \includegraphics[scale=0.4]{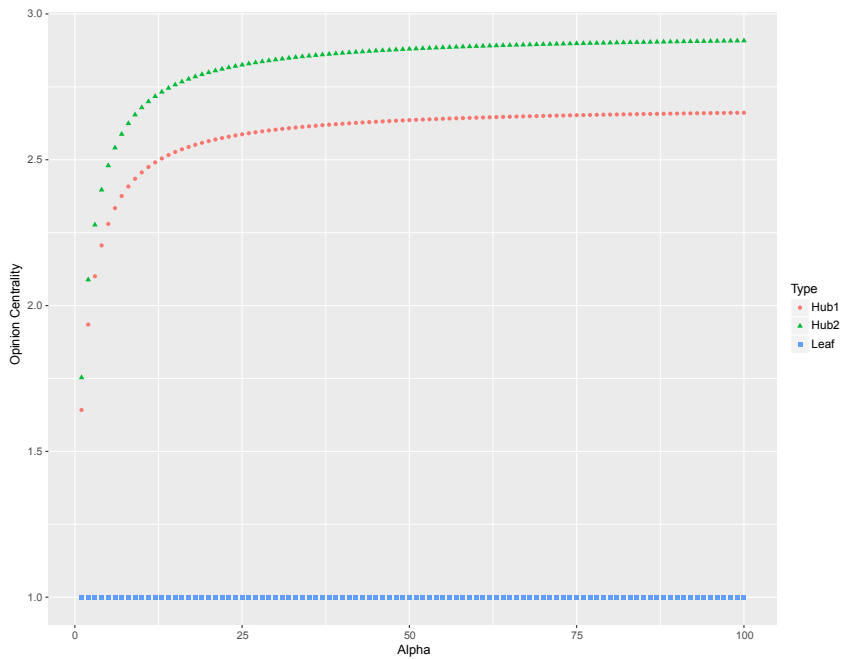}
	\caption{Evolution of the centrality of the node types from Figure~\ref{fig:barrelnetwork}, as a function of $\hat{\alpha}$.}
	\label{fig:centralitybarrelnetwork}
\end{figure}

\subsection{Real-World Data}
\label{sec:RealWorld}
The $20$ real-world networks selected for our evaluation are described in Table~\ref{tab:data}. They were retrieved from Nexus\footnote{\url{http://nexus.igraph.org/}}, the network repository of the igraph library, as well as from M. De Domenico's Web page\footnote{\url{http://deim.urv.cat/~manlio.dedomenico/data.php}}. Note there are two multiplex networks for the Kapferer dataset, corresponding to two successive acquisitions. We selected these networks in order to get data of various sizes, ranging from $10$ to $8215$ nodes, $35$ to $43129$ links, and $2$ to $339$ layers.

For each network, we process the opinion centrality, as well as multiplex variants of classic uniplex measures: Degree, PageRank, Eigenvector, HITS (hub and authority) and Katz centralities. The latter are computed using the software MuxViz\footnote{\url{http://muxviz.net/}} \parencite{Domenico2014a}, which implements tensorial generalizations. Each of them is processed for each layer by taking the multiplex information into account, and an overall measure is obtained by aggregating the resulting values over all layers. 
The \textsc{R} scripts we wrote to process the opinion centrality are publicly available online\footnote{\url{https://github.com/CompNet/MultiplexCentrality}}.

\begin{table}[htb]
	\center
    \setlength{\tabcolsep}{3mm}
		\begin{tabular}{llrr}
			\hline
			Name & \hspace{3mm}Reference & \hspace{3mm}Nodes & \hspace{3mm}Layers \\
			\hline
            Aarhus CS department & \cite{Magnani2013} & 61 & 5 \\
            Arabidopsis GPI & \cite{Domenico2015a} & 6980 & 7 \\
            C. Elegans GPI & \cite{Domenico2015a} & 3879 & 6 \\
            CKM physicians & \cite{Coleman1957} & 246 & 3 \\
            Drosophila GPI & \cite{Domenico2015a} & 8215 & 7 \\
			EU-Air transportation & \cite{Cardillo2013} & 450 & 37 \\
            FAO trade & \cite{Domenico2015a} & 183 & 339 \\
			Hepatitus C GPI & \cite{Domenico2014a} & 105 & 3 \\
            Human-HIV1 GPI & \cite{Domenico2014a} & 1005 & 5 \\
            Kapferer tailor shop & \cite{Kapferer1972} & 39 & 2 \\
			Knoke bureaucracies & \cite{Knoke1981} & 10 & 2 \\
            Lazega law firm & \cite{Lazega2001} & 71 & 3 \\
			London transport & \cite{Domenico2014b} & 369 & 13 \\
			Padgett Florentine families & \cite{Breiger1986} & 16 & 2 \\
            Pierre Auger collaborations & \cite{Domenico2014} & 514 & 16 \\
			Rattus GPI & \cite{Domenico2015a} & 2640 & 6 \\
			Roethlisberger bank & \cite{Roethlisberger1939} & 14 & 6\\
			Sampson monastery & \cite{Breiger1975} & 18 & 8 \\
			Thurmann office & \cite{Thurman1979} & 15 & 2 \\
            Wolfe primates & N/A & 20 & 2 \\
			\hline
		\end{tabular}
	\caption{Multiplex Real-World Networks Selected for the Experimental Assessment}
	\label{tab:data}
\end{table}



We first study how the $\alpha$ parameter described in Section \ref{sec:model} affects the opinion centrality rankings, and then how different these are from those obtained for the other measures. However, the multiplex networks available online do not include this information: only the structure of the network (i.e. the $E_c$ matrices, in the notation of our model). So, we proceed like for the toy problem, and consider $\hat{\alpha}$ values over a wide range: $]0;100]$. Our results show that, as also observed on the toy problem, when $\hat{\alpha}$ increases, the opinion centrality values increases (without considering the budget constraint). We use Spearman's correlation to compare the opinion centrality values processed with all considered $\hat{\alpha}$, and systematically obtain a maximal correlation of $1$. So, it turns out $\hat{\alpha}$ affects the opinion centrality values, but not their rank. The value of this parameter is consequently of no importance if one's objective is to rank the nodes. Note, however, that if $\alpha$ is provided as a part of the input data, it is likely to be heterogeneous (by opposition to the unique $\hat{\alpha}$ we used here as a substitute), and this can lead to ranking differences. In other words, $\alpha$ should be used when available, $\hat{\alpha}$ is only a way to model missing information. 

The values obtained with the opinion measure are generally distributed relatively similarly in the studied networks. These distributions are right-skewed, i.e. most of the nodes have a higher centrality, with respect to the considered network. Moreover, they are distributed relatively homogeneously around a characteristic value, which varies depending on the network. The dispersion around the characteristic value also depends on the data. The networks can be roughly grouped into $4$ variants of this distribution, which are illustrated by Figure~\ref{fig:opiniondistrib}. The distributions obtained for networks Aarhus, Kapferer, Knoke, Lazega, Padget, Sampson, Roethlisberger and Wolfe are similar to the top-left one ; networks CKM Phys. and Thurman to the top-right one ; networks Pierre Auger and London to the bottom-left one ; and networks Arabidopsis, C. Elegans, Drosophila, EU-Air, FAO trade, Hepatitus C, Human-HIV1 and Rattus to the bottom-right one. To ease the comparisons, the opinion centrality  of a node ($x$ axis) is expressed in terms of the \textit{proportion of the budget} it receives in the optimal solution, according to our model.

\begin{figure}[htb]
	\center
	\includegraphics[width=0.4\textwidth]{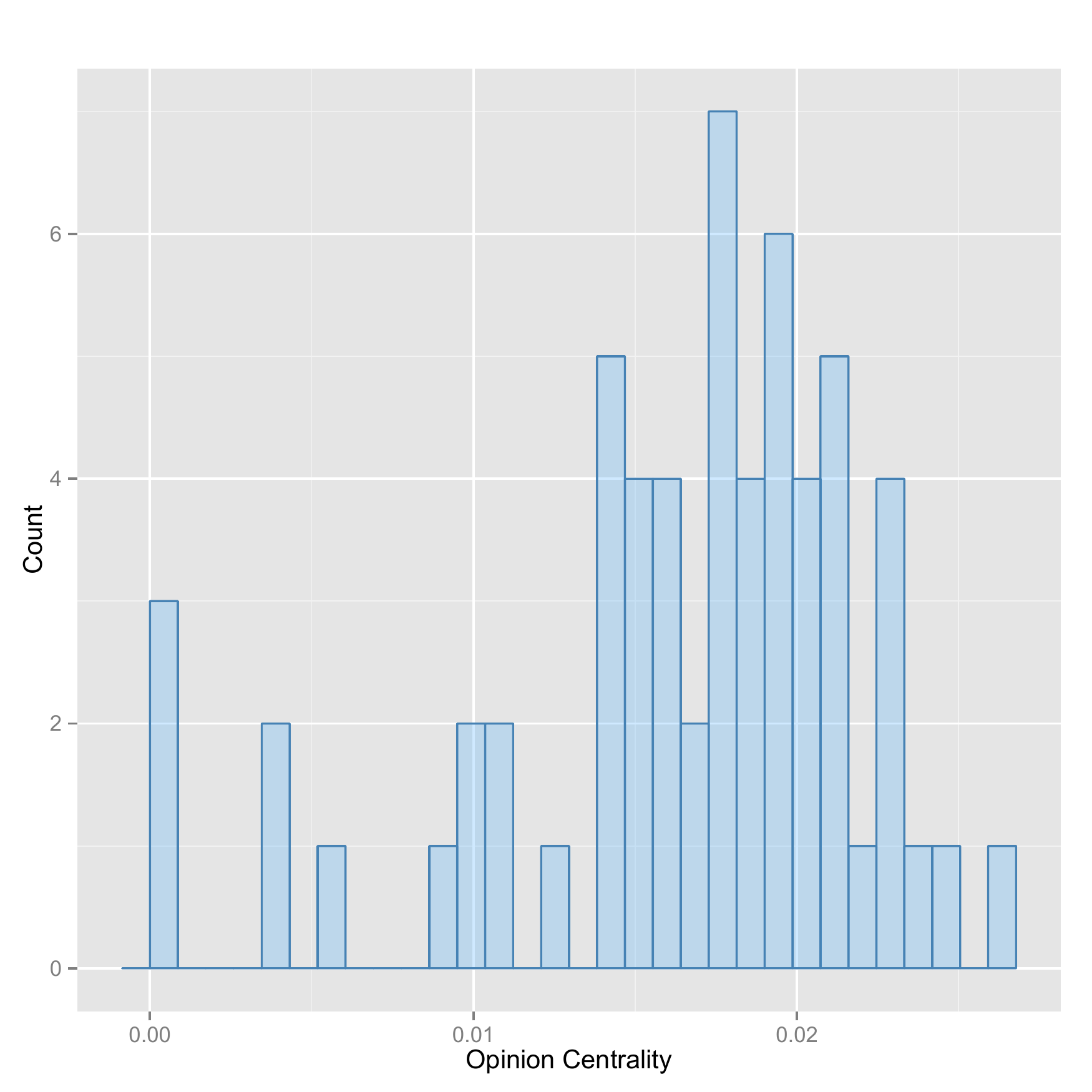}
	\includegraphics[width=0.4\textwidth]{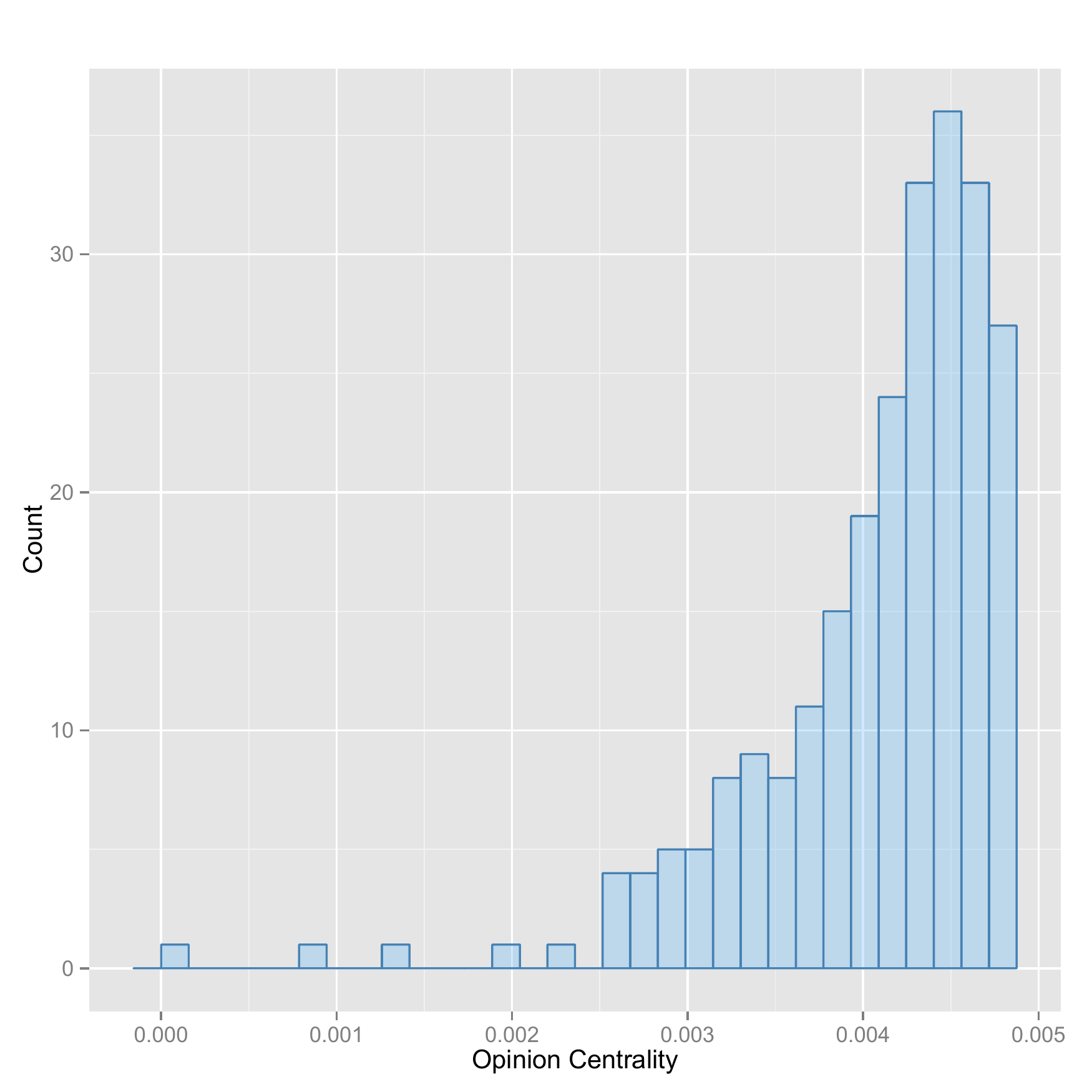}\\
	\includegraphics[width=0.4\textwidth]{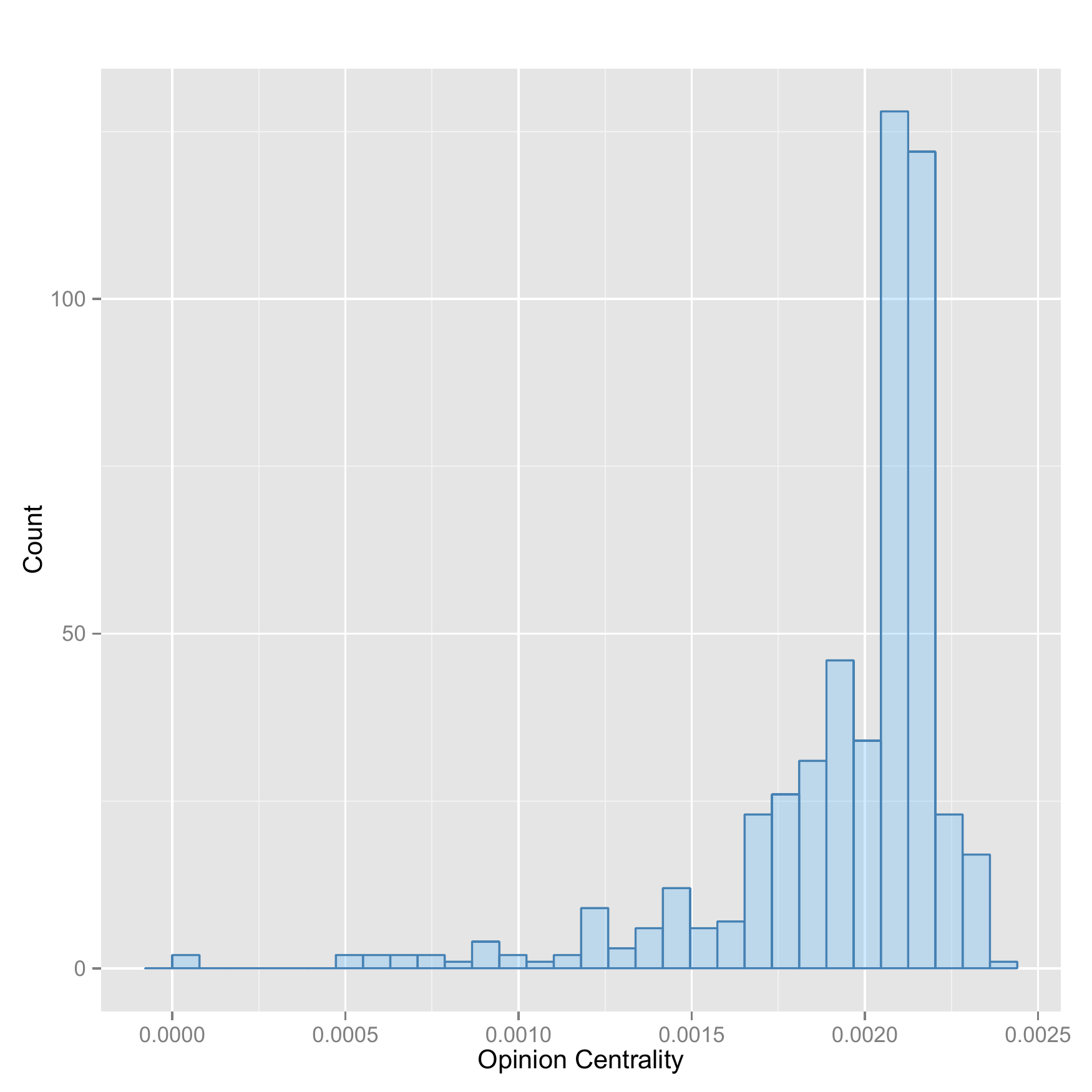}
	\includegraphics[width=0.4\textwidth]{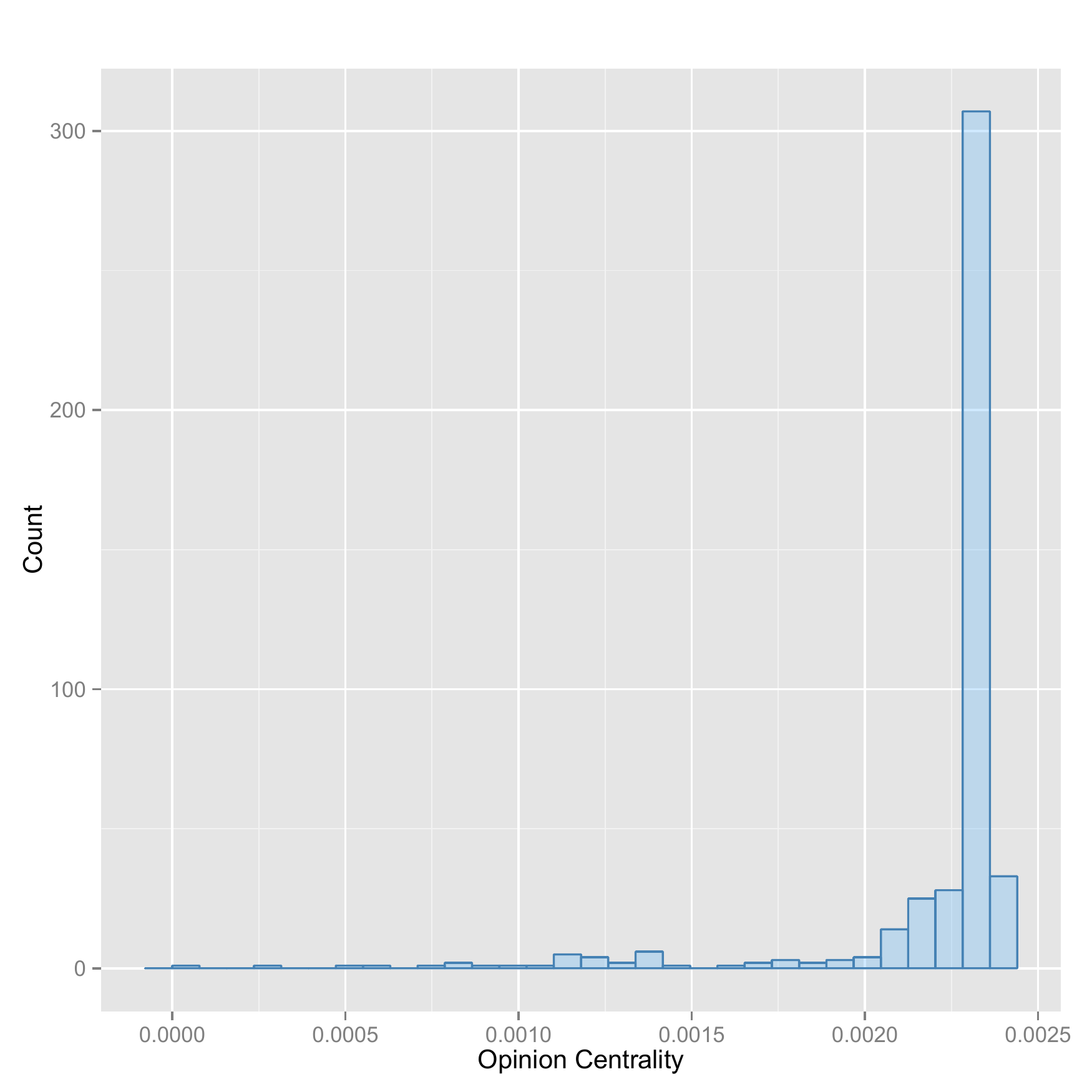}
	\caption{Distribution of the opinion measure for the Aarhus CS department (top-left), CKM Physicians (top-right), Pierre Auger collaborations (bottom left) and EU-Air transportation (bottom right) networks.}
	\label{fig:opiniondistrib}
\end{figure}

We now compare the opinion measure to the other multiplex centrality measures. For each network, we process Spearman's correlation to compare the ranks obtained for the opinion centrality to those of the other measures. The results are displayed in Table~\ref{tab:spearman}. Certain measures (Eigenvector, HITS, Katz, PageRank) could not be processed for the largest networks, due to their memory cost. This point, as well as other computational aspects, are discussed at the end of this section. Missing in- and out-degree values correspond to undirected networks. The correlations vary much depending on both the network and measure of comparison, ranging from $-0.97$ to $0.26$. Overall, we observe a mild to strong negative correlation for all considered measures. However, it is worth noticing that, for a given measure, the magnitude of the correlation varies depending on the network. For instance, it ranges from $-0.96$ to $0.26$ for Katz. This means the opinion measure does not simply systematically reverse the rankings of the considered measure. Moreover, the correlation also varies for a given network, e.g. it ranges from $-0.96$ to $0.02$ for CKM Phys. These observations confirm empirically that the opinion measure characterizes nodes based on different criteria than the other considered centraelity measures, as designed.

\hspace{-2cm}
\begin{table}[htb]
    \setlength{\tabcolsep}{3mm}
    \center
		\begin{tabular}{l@{\hspace{30pt}}r@{\hspace{4pt}}r@{\hspace{4pt}}r@{\hspace{16pt}}r@{\hspace{16pt}}r@{\hspace{4pt}}r@{\hspace{16pt}}r@{\hspace{16pt}}r}
			\hline
			\multirow{2}{*}{Network} & \multicolumn{3}{c}{Degree} & Eigen & \multicolumn{2}{c}{HITS} & \multirow{2}{*}{Katz} & Page \\
			               & Total & In & Out & Vector & Auth. & Hub & & Rank \\
			\hline
			Aarhus & $-0.81$ & $-$ & $-$ & $-0.94$ & $-0.74$ & $-0.74$ & $-0.74$ & $-0.74$\\
			Arabidopsis & $-0.76$ & $-0.45$ & $-0.65$ & $-$ & $-$ & $-$ & $-$ & $-$\\
			Celegans & $-0.38$ & $-0.72$ & $0.13$ & $-0.96$ & $-$ & $-$ & $-$ & $-$\\
			CKM Phys. & $-0.80$ & $-0.96$ & $-0.17$ & $-0.86$ & $-$ & $0.02$ & $-0.49$ & $-0.15$\\
			Drosophila & $-0.73$ & $-0.89$ & $-0.36$ & $-$ & $-$ & $-$ & $-$ & $-$\\
			EU-Air & $-0.95$ & $-$ & $-$ & $-0.96$ & $-$ & $-$ & $-$ & $-$\\
			FAO Trade & $-0.41$ & $-1.00$ & $0.02$ & $-$ & $-$ & $-$ & $-$ & $-$\\
			Hepatitus & $-0.18$ & $0.01$ & $-0.01$ & $-0.57$ & $-$ & $-0.17$ & $-0.49$ & $-$\\
			Human-HIV1 & $-0.48$ & $-0.51$ & $-0.12$ & $-0.72$ & $-0.30$ & $-0.11$ & $0.26$ & $-0.12$\\
			Kapferer1 & $-0.90$ & $-0.91$ & $-0.86$ & $-0.98$ & $-0.78$ & $-0.68$ & $-0.77$ & $-0.66$\\
			Kapferer2 & $-0.90$ & $-0.91$ & $-0.87$ & $-0.95$ & $-0.83$ & $-0.76$ & $-0.82$ & $-0.75$\\
			Knoke & $-0.68$ & $-0.72$ & $-0.34$ & $-0.76$ & $-0.77$ & $-0.16$ & $-0.89$ & $-0.10$\\
			Lazega & $-0.83$ & $-0.93$ & $-0.50$ & $-0.94$ & $-0.90$ & $-0.48$ & $-0.95$ & $-0.48$\\
			London & $-0.72$ & $-$ & $-$ & $-0.77$ & $-0.03$ & $-0.03$ & $-0.03$ & $-0.06$\\
			Padgett & $-0.93$ & $-$ & $-$ & $-0.94$ & $-0.89$ & $-0.89$ & $-0.89$ & $-0.89$\\
			Pierre Auger & $-0.46$ & $-$ & $-$ & $-0.73$ & $-$ & $-0.40$ & $-0.40$ & $-0.47$\\
			Rattus & $-0.50$ & $-0.73$ & $-0.08$ & $-0.88$ & $-$ & $-$ & $-$ & $-$\\
			Roethlisberger & $-0.60$ & $-0.66$ & $-0.59$ & $-0.89$ & $-0.79$ & $-0.73$ & $-0.82$ & $-0.77$\\
			Sampson & $-0.71$ & $-0.93$ & $0.14$ & $-0.82$ & $-0.82$ & $0.00$ & $-0.96$ & $-0.06$\\
			Thurmann & $-0.73$ & $-0.97$ & $-0.70$ & $-0.96$ & $-0.82$ & $-0.71$ & $-0.78$ & $-$\\
			Wolfe & $-0.50$ & $-0.59$ & $-0.45$ & $-0.44$ & $-0.38$ & $-0.36$ & $-0.38$ & $-0.36$\\
			\hline
		\end{tabular}
    \caption{Spearman's correlation between the opinion measure and the other considered multiplex centrality measures.}
	\label{tab:spearman}
\end{table}

To get a better insight of the opinion measure, We consider the nodes individually. The left plot in Figure~\ref{fig:rankdifferences} represents the difference in ranking between the opinion centrality and the Authority measure, on the CKM Phys. network. The right plot is built similarly, but focuses on the out-degree in the Lazega network. Both plots are very typical of what we observe for other data and centrality measures. For a given centrality measure (here: the Authority and the out-degree), the nodes are ordered on the $x$ axis by increasing centrality values, whereas the $y$ axis represent how the node ranking changes from the considered measure to the opinion measure. On the extremes, the opinion measure tends to demote the nodes the most central according to the other measures, whereas it promotes the least central ones. Those all become moderately central, according to the opinion measure. On the contrary, certain nodes previously with intermediate ranking are placed among the most or least central nodes by the opinion measure. By comparison, the same plot built to compare two alternative multiplex measure typically leads to a flatter figure, with smaller ranking differences, especially regarding the most and least central nodes.  

 \begin{figure}[htb]
 	\center
 	\includegraphics[scale=0.49]{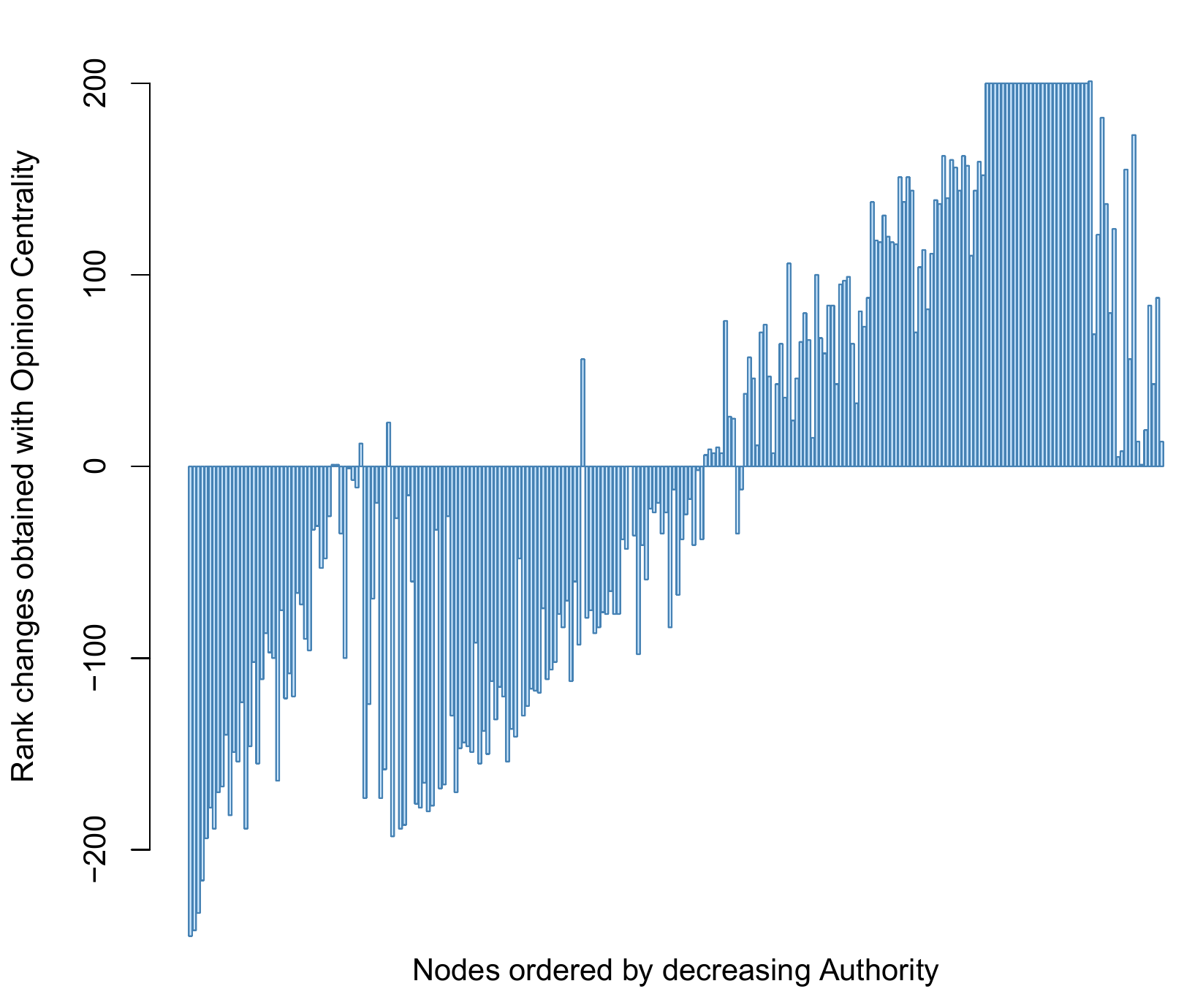}
 	\includegraphics[scale=0.49]{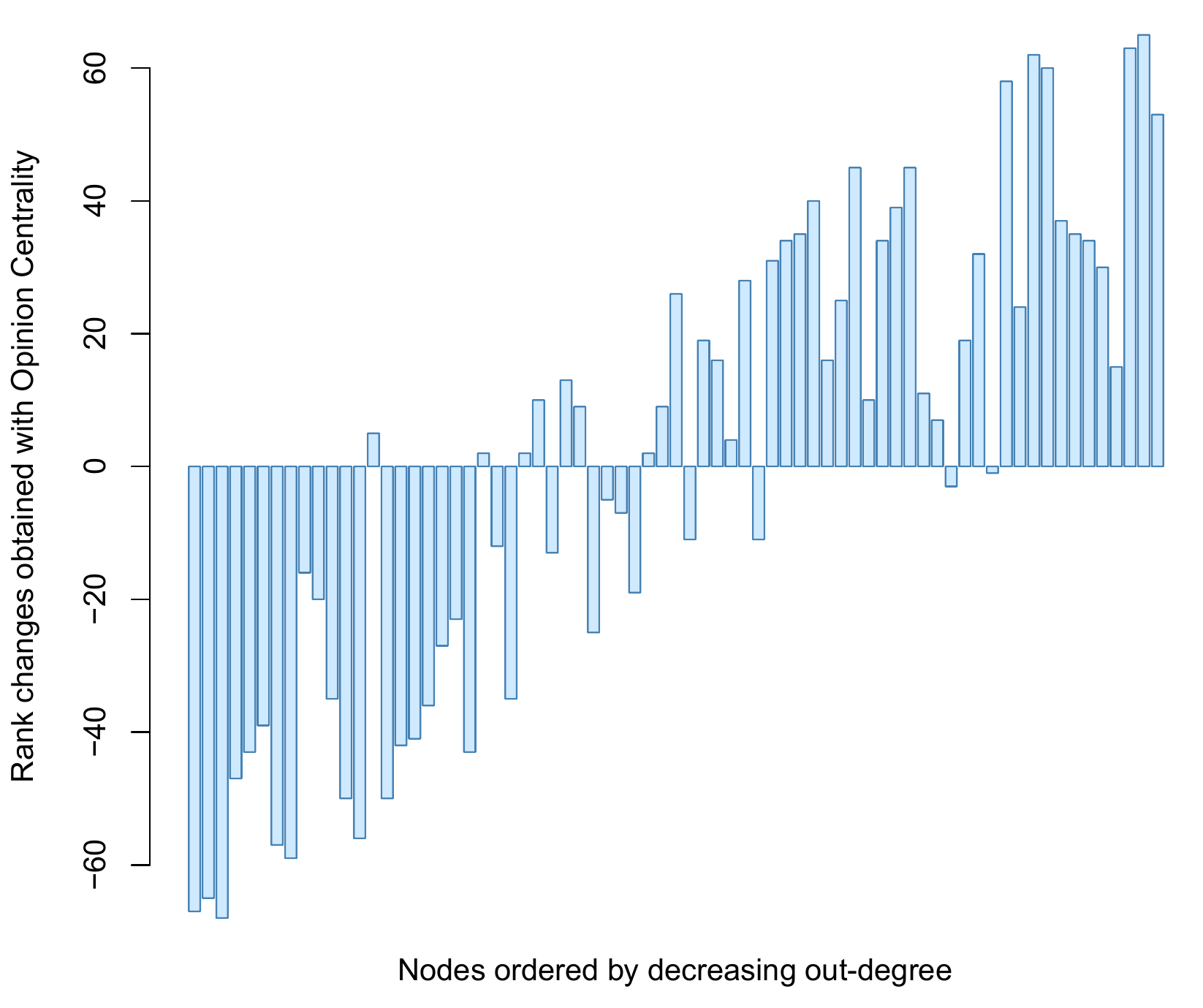}
 	\caption{Rank difference for each individual node, in the CKM Phys. (left) and Lazega (right) networks. The left plot compares the opinion centrality to the Authority measure, whereas it is the out-degree in the right one. Each bar corresponds to a node, and its height matches the rank difference.}
 	\label{fig:rankdifferences}
 \end{figure}

Figure \ref{fig:overallrankdifferences} gives a more global outlook of the opinion centrality behavior relatively to the other measures. Each plot is built on the same principle than the ones in Figure \ref{fig:rankdifferences}, and consequently focuses on a specific alternative measure. However, this time all the networks appear at once, in each plot, and are restricted to their $5$ most central nodes (according to the considered alternative measure). Moreover, to get comparable $y$ values, the rank difference is normalized: we divide it by $n-1$ ($n$ being the number of nodes in each layer) in order to get a value between $-1$ and $+1$. These plots confirm certain of the observations we previously made. In particular, the most central nodes tend to undergo some dramatic ranking changes, but this depends on both the considered measure and network.

\begin{figure}[htb]
	\center
	\includegraphics[scale=0.38]{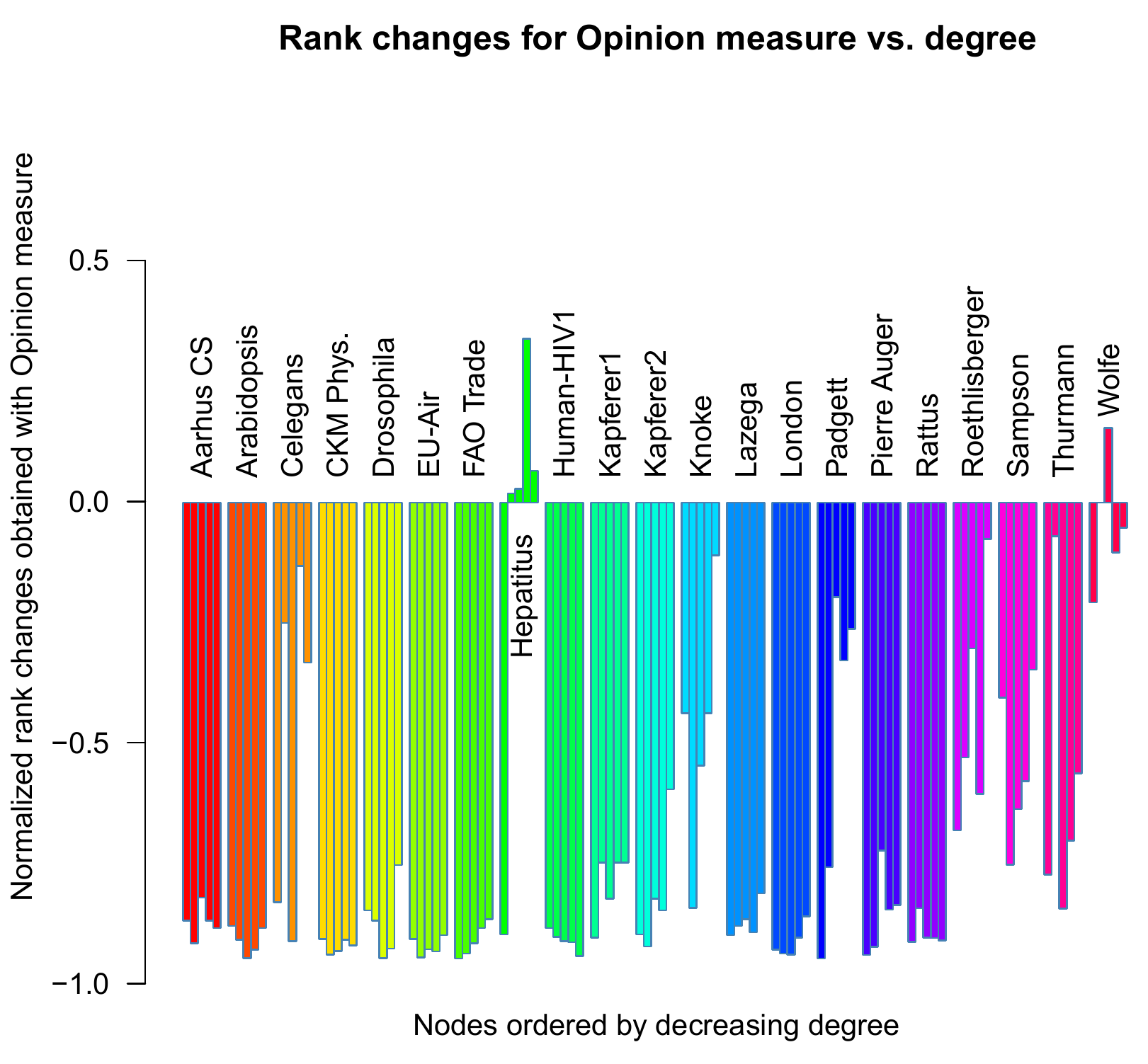}
	\includegraphics[scale=0.38]{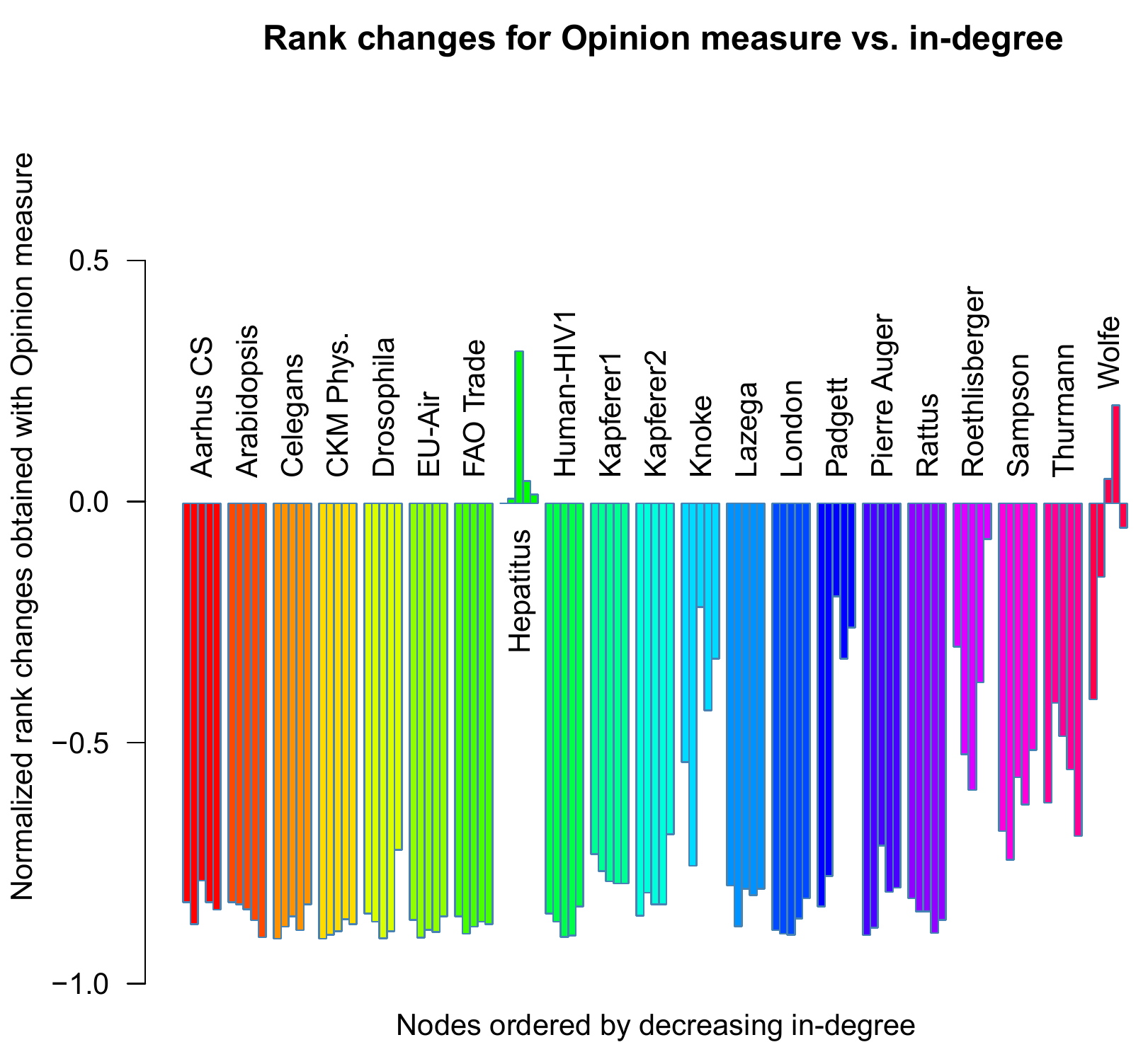}
    \\ \vspace{3mm}
	\includegraphics[scale=0.38]{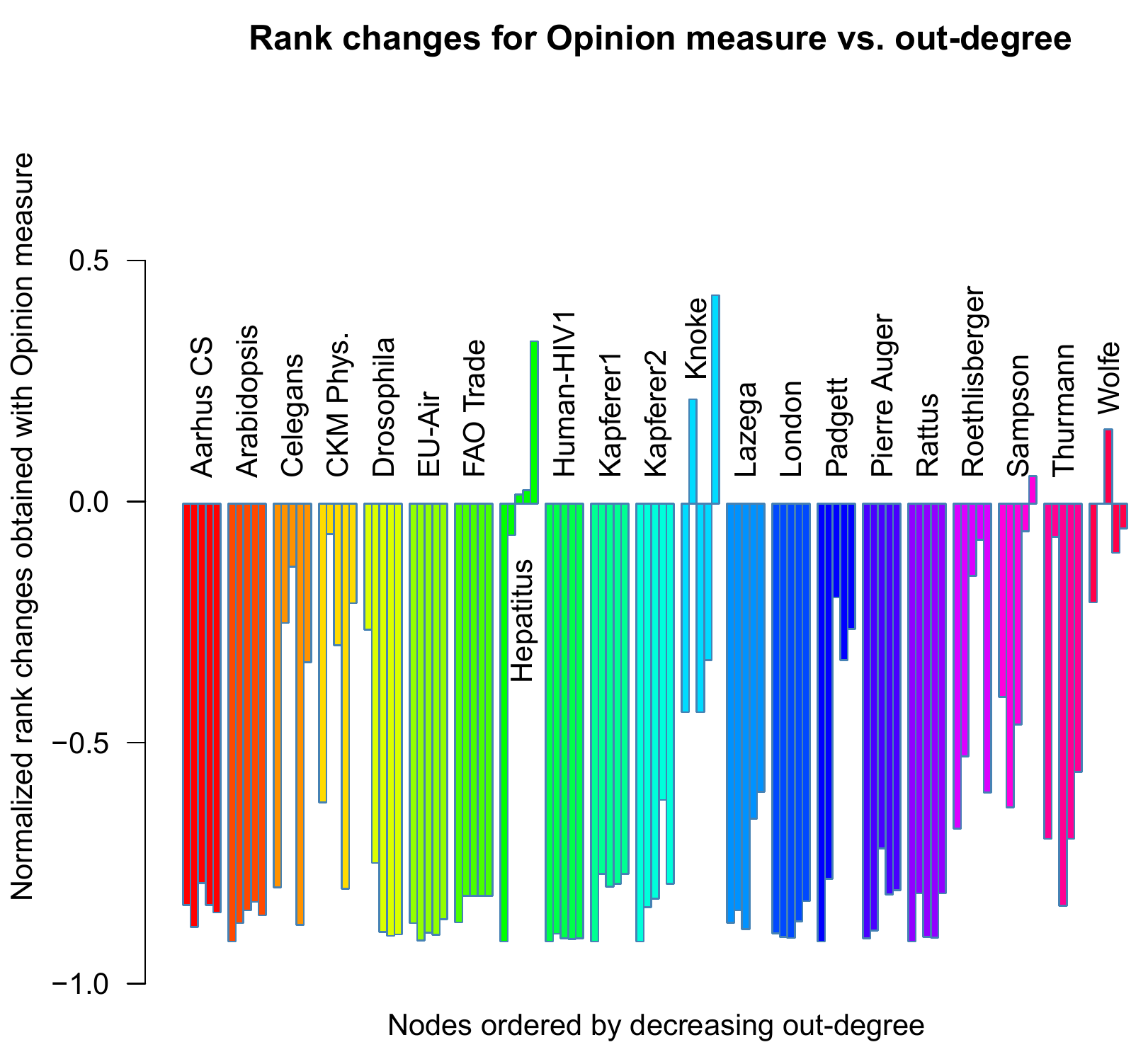}
	\includegraphics[scale=0.38]{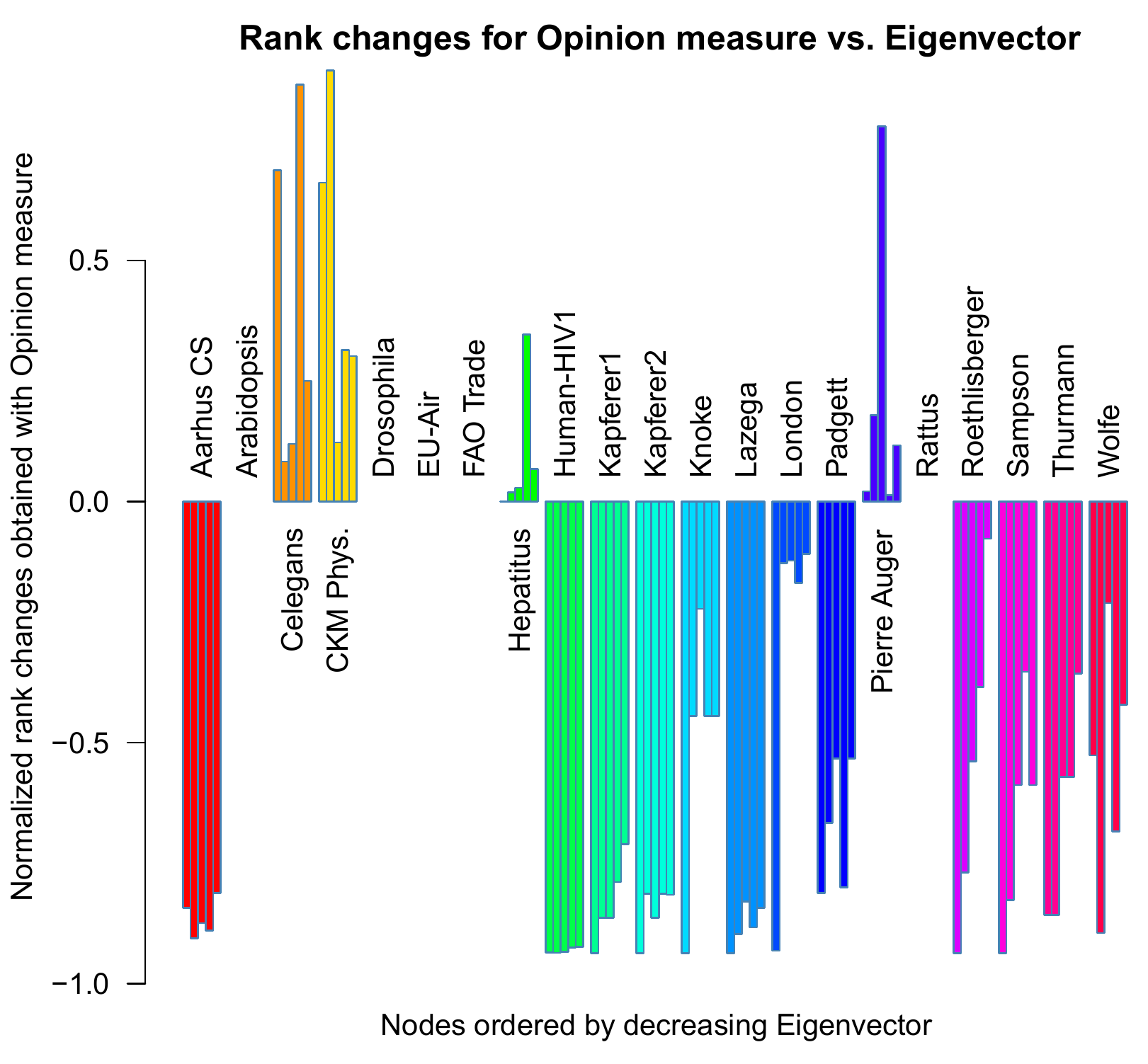}
    \\ \vspace{3mm}
	\includegraphics[scale=0.38]{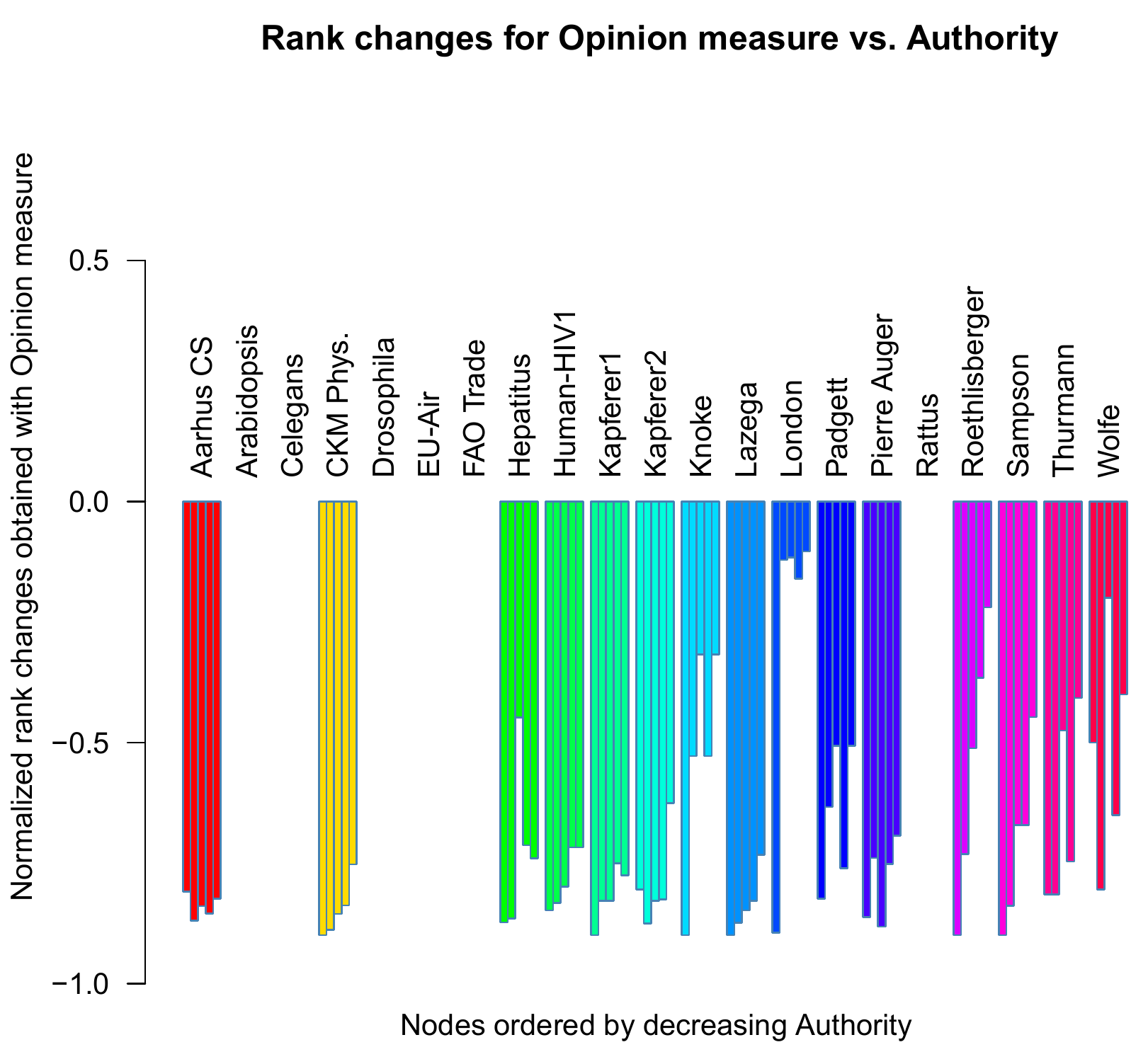}
	\includegraphics[scale=0.38]{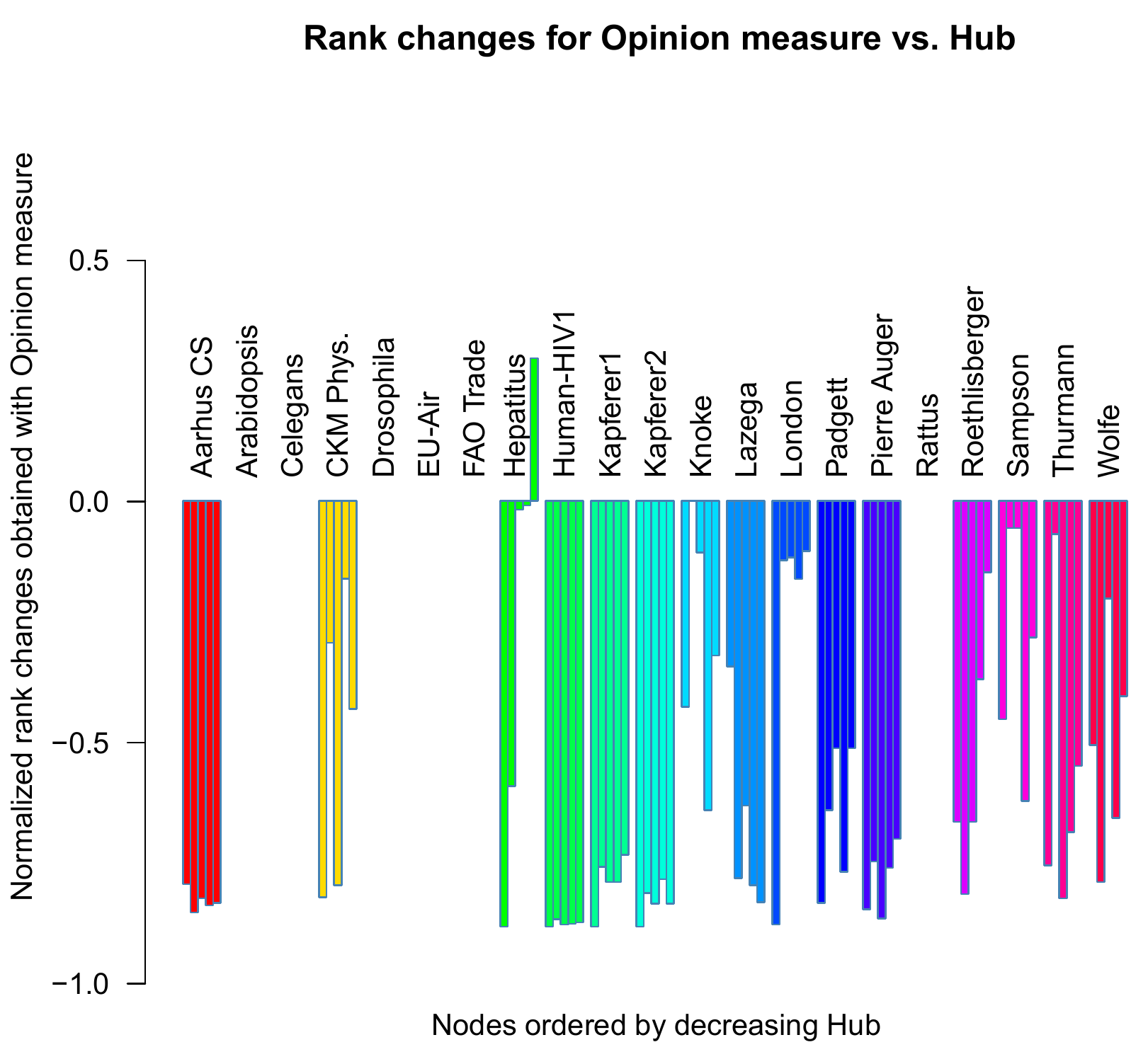}
    \\ \vspace{3mm}
	\includegraphics[scale=0.38]{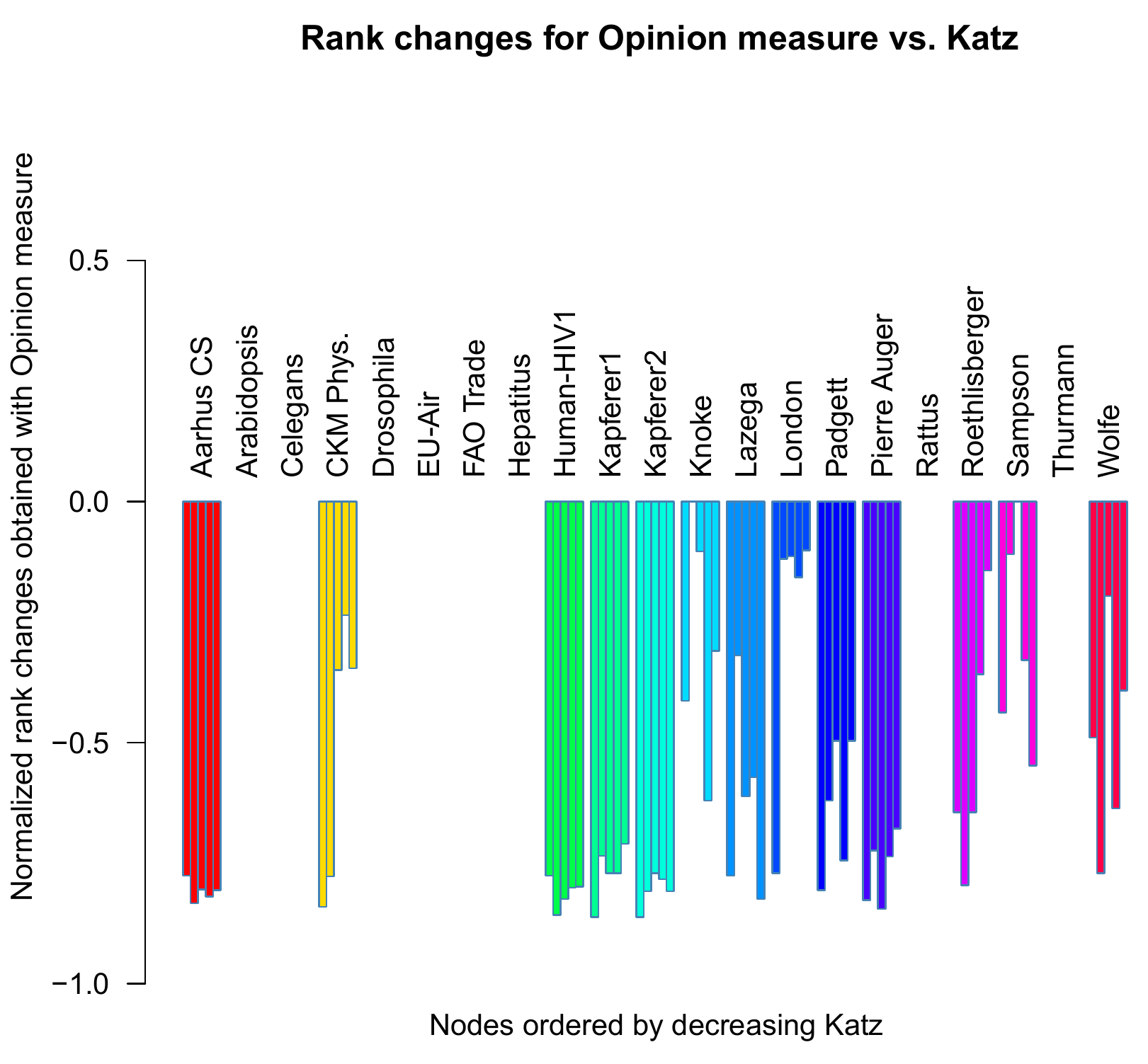}
	\includegraphics[scale=0.38]{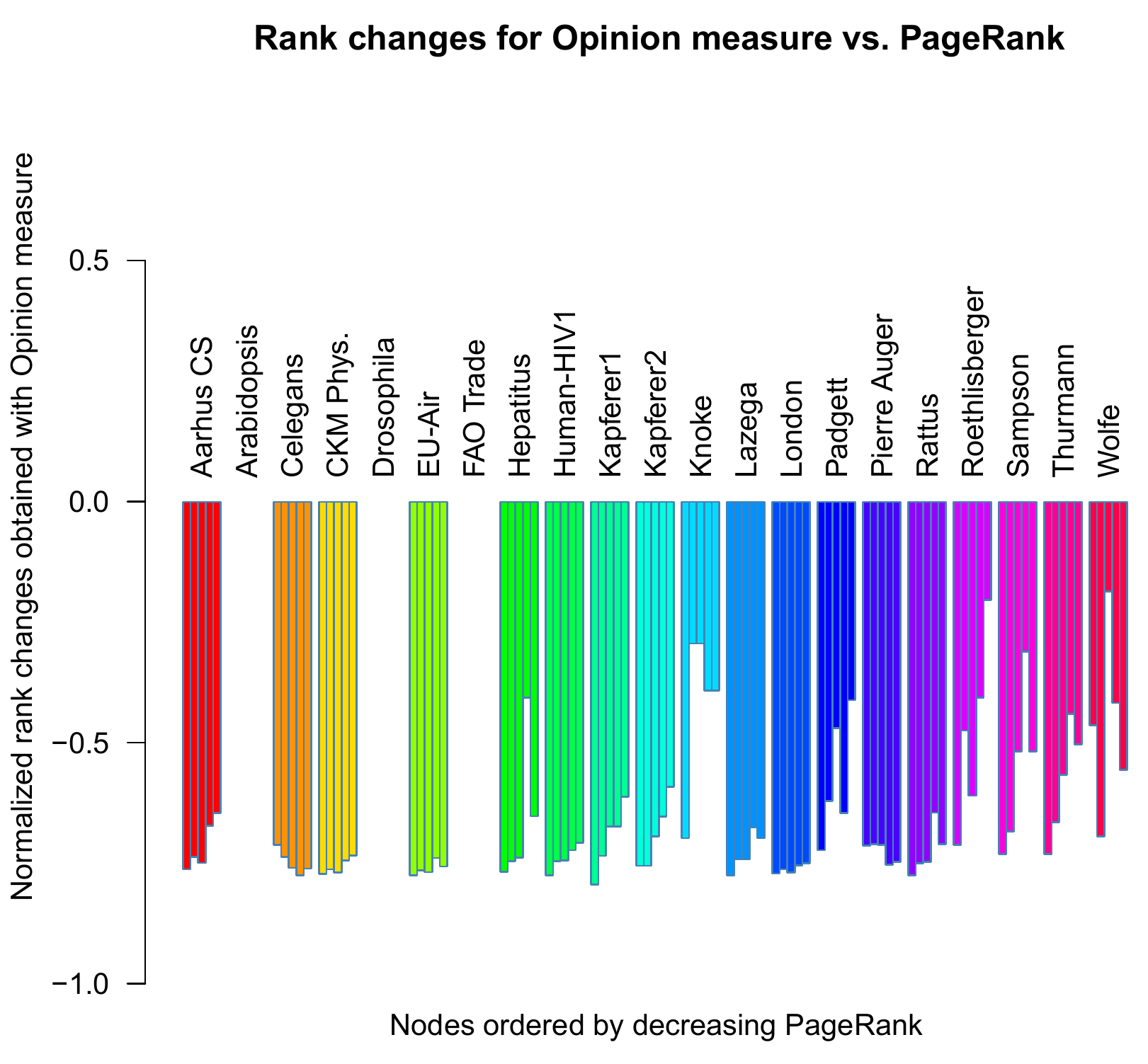}
	\caption{Rank difference for each network. Each plot compares the opinion centrality with one of the alternative multiplex measures considered in this article. The principle is the same than in Figure \ref{fig:rankdifferences}, except only the $5$ most central nodes (according to the alternative measure) are represented for each network, and the $y$ is normalized according to the network size for readability purposes.}
	\label{fig:overallrankdifferences}
\end{figure}

The ranking differences observed between the opinion centrality and the other multiplex measures are due to the optimization problem it is based upon. Indeed, in this problem, it can be necessary to externally stimulate certain nodes which do not have a particularly high degree, or have no neighbors with a particularly high degree. For instance, a leaf node whose unique link is directed towards the rest of the network will not be reachable from another node. So, if one wants to influence the opinion of the whole network, it is worth acting directly on this node. This kind of node is typically considered by measures such as the Degree or Eigenvector centralities as not central. This highlights the fact the semantics of the opinion measure is clearly unlike that of the other measures considered here.

\begin{figure}[!htb]
	\center
	\includegraphics[width=0.4\textwidth]{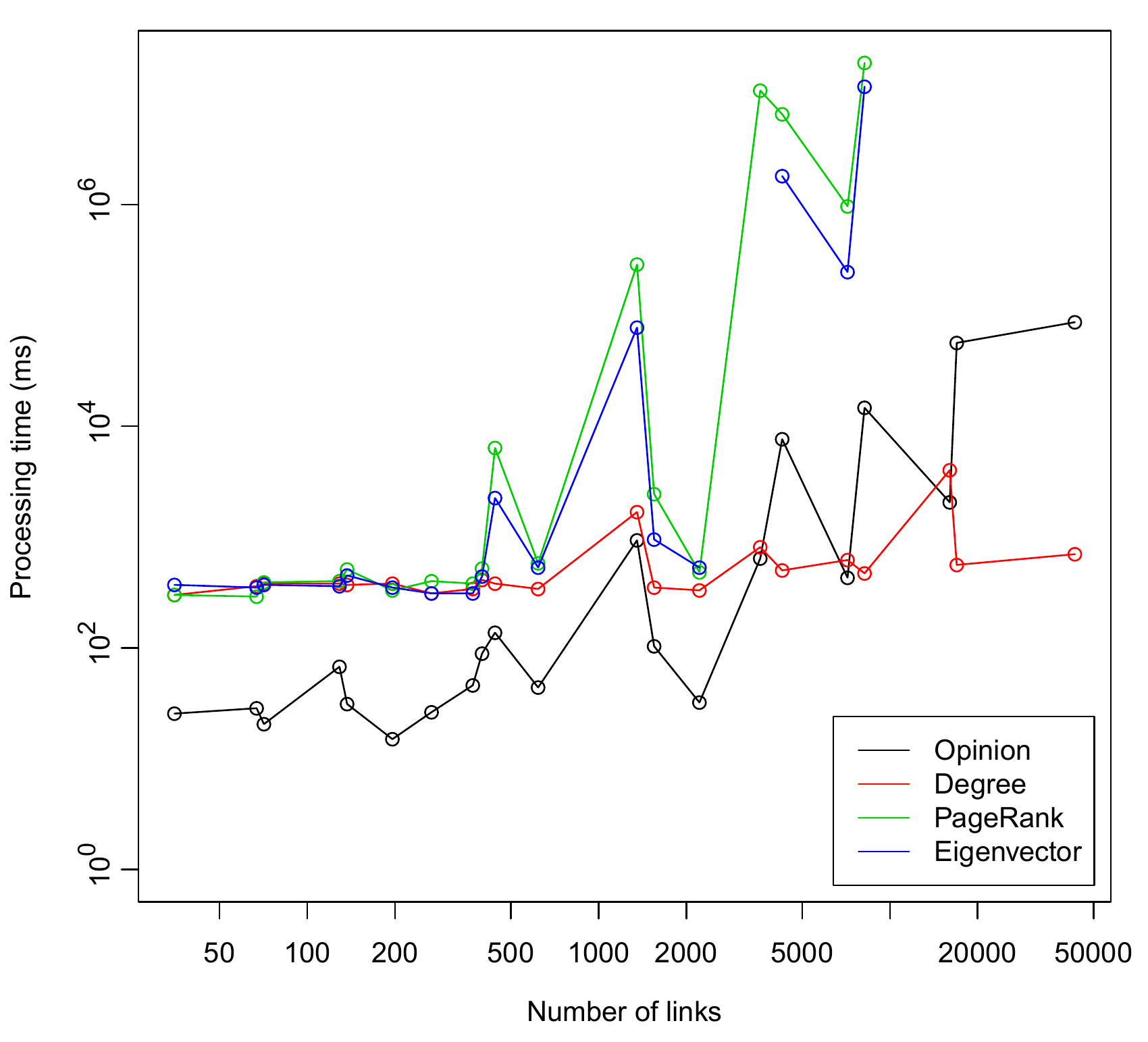}
	\includegraphics[width=0.4\textwidth]{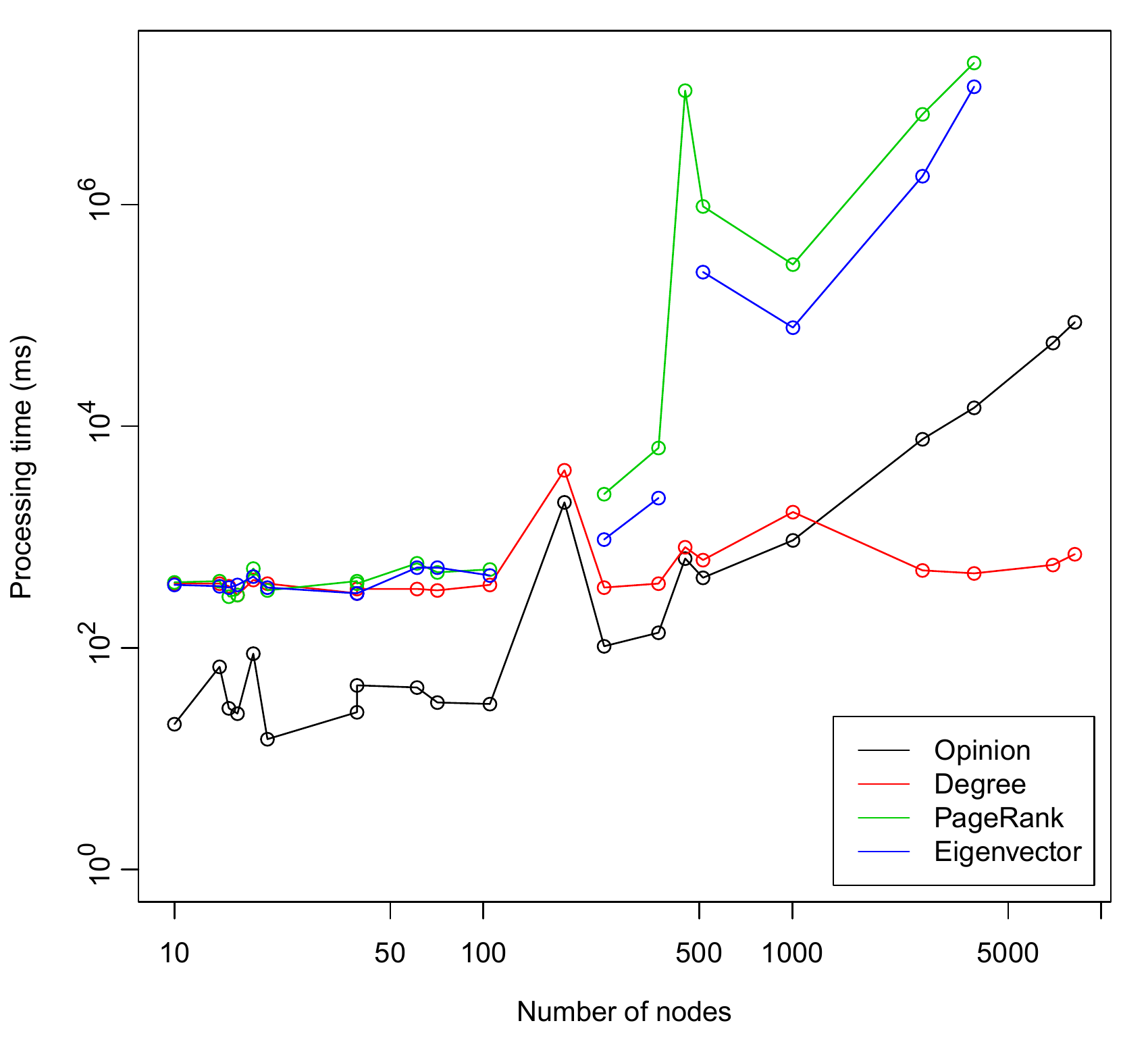}
	\caption{Processing times of the main considered measures, as functions of the numbers of links (left) and nodes (right).}
	\label{fig:processTimes}
\end{figure}

Finally, we compare the computational costs of the measures. During the processing of the opinion centrality, the most expensive operation is the inversion of an $n \times n$ matrix ($n$ being the number of nodes by layer). Therefore the time complexity associated to the opinion centrality is in $O(n^3)$. We do not have access to the algorithmic complexity for the other multiplex measures considered in this article, so we compare the measures empirically. Figure~\ref{fig:processTimes} displays the processing times obtained when computing our opinion centrality, as well as the multiplex versions of the degree, PageRank and Eigenvector measures. We used a plain desktop PC (i5 3.00GHz quadcore processor with 16GB RAM). Note that both axes use a logarithmic scale. The left plot contains the processing times as a function of the number of links in the network, whereas the right one focuses on the number of nodes. We did not include all measures for readability matters, and because in- and out-degree behave like degree, whereas Hub, Authority and Katz performances are located in between Eigenvector and PageRank. 

The processing time for the degree is quite stable, as expected from this purely local measure. For the Eigenvector and PageRank measures, it increases exponentially with both the number of nodes and links in the network. The opinion measure also undergoes a very fast increase, but clearly slower than Eigenvector and PageRank. For the largest network, it is a matter of minutes. In terms of memory usage, the opinion centrality is also less expensive, as illustrated by the fact we could not process the Eigenvector, Authority, Hub, Katz and PageRank measures for the largest networks considered in this study, due to memory limits. Finally, we did not detect any effect of the opinion centrality parameters on its processing time.

\section{Conclusion}
In this article, we presented the opinion centrality, a measure designed to characterize the relative position of nodes in a multiplex network. Our work relies on a stochastic model representing opinion diffusion dynamics in a social group of persons communicating through several independent media. The opinion centrality is derived from the solution of an optimization problem defined on this model, and consisting in maximizing the overall opinion of the social group through direct individual influence. We show on a toy example and a collection of real-world networks that the node rankings obtained with the opinion centrality clearly differ from other multiplex measures. In particular, high degree nodes are not necessarily considered as central, and low degree nodes can be seen as central if they allow a better control of the opinion diffusion.

Our work can be extended in various ways. First, we limited our experimental assessment to uniform internal influence coefficients (parameter $\alpha$ in our model), because the available real-world data do not provide this type of information, and therefore do not allow any validation regarding the use of non-uniform values. However, it would be interesting to explore this trail, possibly using artificially generated networks. On the same note, we also used the simplest available utility function (plain sum), but we want to explore how the measure behaves when using more advanced utilities such as the ones we mentioned earlier in the article (weighted sum, minimum, product). Second, to ease human interpretation and to allow comparisons with other multiplex measures, we focused our tests on small and medium-sized real-world networks only. The next step will consist in studying how the opinion measure scales when applied to much larger networks, such as those provided with MuxViz. 

Third, regarding the model itself, we can think of two promising extensions. The first one is the temporal network extension. There are several very different ways of modeling opinion dynamics in a temporal network. For instance, it is necessary to decide if the network structure and the opinion dynamics evolve simultaneously. Concerning the opinion centrality, its calculation could not rely on static control, like in this paper, and one should rather use the theory of optimal control. The other extension is related to the amount of information available when processing the opinion centrality. It would be interesting to consider the case where we do not know the structure of the network, and only get some feedback from the users. This problem could be solved by using an approximate gradient algorithm. In the worst case scenario, there would be no feed back at all, and we could use matrix completion techniques coupled with gradient optimization in order to recover the opinion centrality.

\printbibliography

\vspace{1cm}


\end{document}